\documentclass[12pt]{arxiv}

\usepackage{qcircuit}

\usepackage{amsmath}
\usepackage{amsfonts}
\usepackage{mathrsfs}
\usepackage{tikz}
\usepackage[utf8]{inputenc}

\usepackage{epstopdf}
\usepackage{mathtools}
\usepackage{graphicx}
\usepackage{bbm}
\usepackage{enumerate}
\usepackage{epsf,verbatim,amssymb,array,cite,multicol,multirow}  
\usepackage{psfrag,bm,xspace}
\usepackage{hhline}
\usepackage{xstring}
\usepackage{ifthen}
\usepackage{booktabs} 

\usepackage{makecell}

\usepackage{xparse}
\usepackage{physics}

\ifdefined \theorem 
\else
  \newtheorem{theorem}{Theorem}
\fi
\ifdefined \lemma 
\else

\newtheorem{lemma}{Lemma}
\fi

\ifdefined \corollary 
\else
\newtheorem{corollary}{Corollary}
\fi

\newtheorem{fact}{Fact}

\ifdefined \proposition 
\else
\newtheorem{proposition}{Proposition}
\fi

\ifdefined \definition 
\else
\newtheorem{definition}{Definition}
\fi

\ifdefined \example 
\else
\newtheorem{example}{Example}
\fi

\ifdefined \remark 
\else
\newtheorem{remark}{Remark}
\fi

\ifdefined \conjecture 
\else
\newtheorem{conjecture}{Conjecture}
\fi

\makeatletter
\def\old@comma{,}
\catcode`\,=13
\def,{%
  \ifmmode%
    \old@comma\discretionary{}{}{}%
  \else%
    \old@comma%
  \fi%
}
\makeatother

\newcommand\numberthis{\addtocounter{equation}{1}\tag{\theequation}}

\usepackage{xcolor}
\definecolor{darkblue}{rgb}{0.1,0.1,0.8}
\definecolor{DarkGreen}{rgb}{0,0.6,0}
\definecolor{brickred}{rgb}{0.8, 0.25, 0.33}
\definecolor{britishracinggreen}{rgb}{0.0, 0.26, 0.15}
\definecolor{calpolypomonagreen}{rgb}{0.12, 0.3, 0.17}
\definecolor{ao(english)}{rgb}{0.0, 0.5, 0.0}
	\definecolor{cadmiumgreen}{rgb}{0.0, 0.42, 0.24}
\definecolor{burgundy}{rgb}{0.5, 0.0, 0.13}
\usepackage{etoolbox}

\providetoggle{Blue_revision}
\settoggle{Blue_revision}{true}

\providetoggle{Track}
\settoggle{Track}{true}
\newcommand{\addv}[3]{%
	\iftoggle{Track}{%
    	\IfEqCase{#1}{%
       	 	{a}{\ifthenelse{\equal{#2}{ON}}{{\color{cadmiumgreen}#3}}{#3}}%
        	{b}{\ifthenelse{\equal{#2}{ON}}{{\color{brickred}#3}}{#3}}%
       		{c}{\ifthenelse{\equal{#2}{ON}}{{\color{burgundy}#3}}{#3}}%
    	}[\PackageError{tree}{Undefined option to tree: #1}{}]%
	}{#3}%
}

 \usepackage{hyperref}

\hypersetup{
colorlinks=true, %
  pdfstartview={FitH},
    linkcolor=red,
    citecolor=blue,
    urlcolor={blue!80!black}
}

\usepackage{graphicx}

\newcounter{relctr} 
\everydisplay\expandafter{\the\everydisplay\setcounter{relctr}{0}} 

\AtBeginDocument{} 

\global\long\def\II{\mathbb{I}}
\global\long\def\RR{\mathbb{R}}

\global\long\def\11{\mathbbm{1}}

\newcommand{\bfb}{\mathbf{b}}

\newcommand{\bfs}{\mathbf{s}}
\newcommand{\bft}{\mathbf{t}}
\newcommand{\bfu}{\mathbf{u}}
\newcommand{\bfv}{\mathbf{v}}

\newcommand{\bfx}{\mathbf{x}}

\newcommand{\bfX}{\mathbf{X}}

\newcommand{\CA}{\mathcal{A}}

\newcommand{\ClC}{\mathcal{C}}

\newcommand{\CE}{\mathcal{E}}
\newcommand{\CF}{\mathcal{F}}

\newcommand{\CM}{\mathcal{M}}

\newcommand{\CO}{\mathcal{O}}

\newcommand{\CS}{\mathcal{S}}

\newcommand{\CU}{\mathcal{U}}

\newcommand{\CX}{\mathcal{X}}

\global\long\def\+{\oplus}

\def\<{\langle}
\def\>{\rangle}

\newcommand{\AND}{\mathsf{AND}}

\ifdefined \var 
  \renewcommand{\var}{\mathsf{var}}
\else
  \newcommand{\var}{\mathsf{var}}
\fi

\ifdefined \abs \else
 \newcommand{\abs}[1]{\lvert#1\rvert}
\fi

\ifdefined \norm \else
 \newcommand{\norm}[1]{\lVert#1\rVert}
\fi

\ifdefined \set 
  \renewcommand{\set}[1]{\left\{#1\right\}}
\else
  \newcommand{\set}[1]{\left\{#1\right\}}
\fi

\DeclareMathOperator*{\argmax}{arg\,max}
\DeclareBoldMathCommand{\Re}{Re}
\usepackage{stackengine}

\newcommand{\phs}{\phi_{\bfs}}

\usepackage[nohyperlinks]{acronym}

\newacro{ML}[ML]{machine learning}
\newacro{IID}[i.i.d.]{independent and identically distributed} 
\newacro{PAC}[PAC]{\textit{probably approximately correct}}
\newacro{LCU}[LCU]{Linear Combination of Unitaries}

\newtheorem{observation}{Observation}

\def\L1{{L_1}}

\DeclareMathOperator*{\tensor}{\otimes}
\DeclareMathOperator{\poly}{poly}

\def\qps{\sigma^\bfs}

\def\qas{a_{\bfs}}

\def\pauliset{\{0,1,2,3\}}
\def\paulisetn{\{0,1,2,3\}^n}
\newcommand{\avec}[2]{\norm{#1}_{#2,P}}
\newcommand{\defeq}{\mathrel{:\mkern-0.25mu=}}
\usepackage{breakcites}

\ifdefined \result 
\else
  \newtheorem{result}{Result}
\fi

\title[Learning nearly sparse unitaries ]{Query Learning Nearly Pauli Sparse Unitaries in Diamond Distance}
\usepackage{times}

\coltauthor{
 \Name{Zahra Honjani} \Email{zhonjani@iu.edu}\\
 \addr Department of Computer Science, Indiana University Bloomington
 \AND
 \Name{Mohsen Heidari} \Email{mheidar@iu.edu}\\
 \addr Department of Computer Science, Indiana University Bloomington
}

\begin{document}

\maketitle

\begin{abstract}
    We  study the problem of learning nearly $(s,\epsilon)$-sparse unitaries, meaning that the Pauli spectrum is concentrated on at most  $s$ components with at most $\epsilon$ residual mass in Pauli $\ell_1$-norm.   This class generalizes  well-studied families, including  sparse unitaries, quantum $k$-juntas, $2^k$-Pauli dimensional channels, and compositions of depth $O(\log\log n)$ circuits with near-Clifford circuits. 
  Given query access to an unknown nearly  sparse unitary $U$, our goal is to efficiently (both in time and query complexity) construct a quantum channel that is close in diamond distance to $U$. We design a learning algorithm achieving  this guarantee  using $\tilde{O}(s^6/\epsilon^4)$ forward queries to $U$, and running time polynomial in relevant parameters.    

A key contribution is  an efficient quantum algorithm that, given query access to an arbitrary unknown unitary $U$, estimates  all Pauli coefficients (up to a shared global phase) whose magnitude exceeds a given threshold $\theta$, extending existing sparse recovery techniques to general unitaries.
    
  We also study the broader class of unitaries with bounded Pauli $\ell_1$-norm.  For that class, we prove an exponential query lower  bound $\Omega(2^{n/2})$. We introduce a more relaxed accuracy metric which is the diamond distance restricted to a set of input states. Then,  we show that, under this metric,  unitaries with Pauli $\ell_1$-norm uniformly bounded by $L_1$ are learnable with $\tilde{O}(L_1^8/\epsilon^{16})$. 
\end{abstract}

\section{Introduction}\label{sec:intro}
Understanding the foundations of learning unknown quantum processes is a central problem with direct implications for the characterization and verification of quantum systems. This problem underlies a broad range of tasks, including quantum process tomography, Hamiltonian learning, and the verification and benchmarking of quantum devices.

Without structural assumptions, however, learning quantum processes is intractable. In particular, learning an unknown  $n$-qubit unitary with  additive error $\epsilon$  in diamond distance  requires $\Omega(4^n/\epsilon)$ queries \citep{Haah2023}. This raises a fundamental  question: \textit{ what natural structural assumptions and accuracy metrics enable efficient learning of quantum processes, and what are the resulting sample complexity bounds ?}

A growing body of work has addressed this question by studying unitaries with locality or sparsity structures in a known basis; most prominently the Pauli basis. Pauli sparsity leads to a simple description in a physically meaningful operator basis,  and parallels classical  sparse Boolean functions \citep{KM1993}. 

Formally, an $n$-qubit operator $U$ can be expressed in the Pauli basis as $U = \sum_{\mathbf{s}} \alpha_{\mathbf{s}} \sigma^{\mathbf{s}}$, where $\sigma^{\mathbf{s}}$ are $n$-qubit Pauli operators indexed by $\mathbf{s} \in \{0,1,2,3\}^n$ and $\alpha_{\mathbf{s}}$ are the corresponding Pauli coefficients. 
An operator is  $s$-sparse  if it has at most $s$ non-zero Pauli coefficients.

In this work, we study the more general class of  nearly $(s, \epsilon)$-sparse unitaries, whose Pauli spectrum is concentrated on a set $\CS$ of size at most $s$  dominant components, while allowing for a $\epsilon$ small residual mass, i.e.,  $\sum_{\bfs\notin \CS} |\alpha_\bfs| \leq \epsilon$.

Nearly sparse unitaries generalize several  families including exactly $s$-sparse, quantum $k$-juntas,  and compositions of depth $O(\log\log n)$ circuits with near-Clifford circuits that have been studied extensively. Particularly,  
quantum $k$-juntas, as a special case of $4^k$-sparse unitaries, are learnable with $O(4^k/\epsilon)$ queries \citep{Chen2022junta,Bao2023,Montanaro2010,Grewal2025,Heidari2023a}. More generally, unitaries whose Pauli support lies within a Pauli subgroup of size $2^k$, and are therefore $2^k$-sparse, can be learned with optimal query complexity $\Theta(2^k/\epsilon)$ \citep{Grewal2025}. Learning general sparse (and more broadly nearly sparse) unitaries  remains an open question.

This work aims to fill this gap by studying the learnability of nearly sparse unitaries. We make three main contributions, summarized bellow. 

\subsection{Summary of the main results}

First, we propose an efficient algorithm for identifying and estimating large Pauli coefficients of an unknown unitary $U$ using only forward query access to $U$.
This leads to the following result, abbreviated from Theorem \ref{th:pauli_estimate} in Section \ref{sec:Pauli est}:

\begin{result}
There is an algorithm that, given $\theta$ and query access to an unknown unitary $U$, estimates all Pauli coefficients of $U$ with magnitude at least $\theta$, up to a global phase factor, with additive error $O(\tfrac{\epsilon}{\theta})$ using $\tilde{O}(\frac{\log(1/\theta)}{\epsilon^4})$ queries. 
\end{result}
The algorithm applies to arbitrary unitaries, without any sparsity assumptions, and treats $U$ purely as a black box accessed via forward queries.
This improves upon existing approaches \citep{Arunachalam2025,abbas2025nearly,Montanaro2010,Grewal2025,hu2025ansatz}, which require additional structure such as subgroup support, Boolean-valued spectra, or access to time-evolution at small evolution times. 

Note that the resulted estimates $\hat{\alpha}_\bfs$ are close to $e^{i\phi} \alpha_\bfs$, for some unknown $\phi \in [0, \pi]$. This global phase ambiguity is unavoidable without access to controlled-$U$. 

Our approach builds upon the estimation algorithm of \citep{Arunachalam2025} originally developed for Hamiltonian learning via access to short-time evolution operators of the form $U = e^{-iHt}$. We employ Bell sampling on the Choi state of $U$ to efficiently locate the dominant Pauli components.

However, the coefficient estimation techniques used in prior Hamiltonian learning works \citep{Arunachalam2025,abbas2025nearly} do not directly extend to arbitrary unitaries, as they rely on linear approximations of the form $U \approx I - itH$, which require small evolution times.
In contrast, a general unitary or the time evolution of a Hamiltonian at large times—need not be close to the identity.

Moreover, existing shadow tomography methods for estimating Pauli observables, such as \citep{Chen2024,King2025}, are not directly applicable, since the Pauli coefficients of a unitary operator are generally complex-valued rather than real.
We overcome these challenges by first estimating the magnitude of the largest Pauli coefficient, and then using this estimate to combine Bell sampling with classical shadow tomography based on random Clifford measurements.
This enables accurate estimation of both the magnitude and relative phase of the dominant Pauli coefficients of $U$.

Second, we analyze the query complexity of learning nearly sparse unitaries. Moreover, we use our estimation algorithm for efficiently learning nearly-sparse unitaries.  We have the following main results, abbreviated from Theorem \ref{th:learning_n_sparse} in Section \ref{sec:Ulearning}:
\begin{result}
There exists an algorithm that learns nearly $(s,\epsilon)$-sparse unitaries in diamond distance using $\tilde{O}(\frac{s^6}{\epsilon^4 })$ queries. Moreover, there exists a polynomial time quantum algorithm that learns nearly sparse unitaries using the same number of queries, and with additive error $O(\sqrt{s\epsilon})$ in diamond distance.
\end{result}
Learning unitaries under the diamond norm poses several challenges. The diamond distance depends on the $\ell_1$-norm of the Pauli coefficient vector, necessitating accurate estimation of all coefficients.   When the target unitary is $s$-sparse, it suffices to estimate the $s$ nonzero coefficients to $\ell_\infty$ accuracy. In contrast, for nearly sparse unitaries, one must contend with estimating a large number of small magnitude coefficients. Moreover, truncating to the largest Pauli coefficients does not, in general, yield a valid unitary operator, and any approximation must also admit an efficient implementation in polynomial time. 

We overcome these obstacles by leveraging the \ac{LCU} framework \citep{Berry2014}, which allows us to construct an efficiently implementable approximation of the target unitary directly from the estimated Pauli coefficients.

Third, we study learnability of an even broader class of unitaries; those for which the $\ell_1$-norm of the Pauli coefficients is bounded, i.e., $\sum_\bfs |\alpha_\bfs| \le L_1$.  We establish a lower bound showing that learning unitaries with bounded Pauli $\ell_1$-norm is intractable in the worst case. 

This implies that the diamond distance is too strong for learning this class. 

Motivated by this limitation, we next ask whether there exist alternative accuracy metrics for learning unitaries that both admit meaningful operational interpretations and can be efficiently estimated from experimental data. The diamond distance is the  worst-case trace distance between channel outputs over all possible input states. 

However, such a worst-case requirement may be unnecessarily strong in practical settings, where the learned unitary is only applied to a restricted family of states.

We introduce a restricted diamond distance that quantifies distinguishability only with respect to a subset of states specified based on the learning task. In a similar spirit,  \citet{abbas2025nearly} adopted the time (temperature)-concentrated distance as the metric for Hamiltonian learning, relaxing the worst case distinguishability. 

In particular, we define the uniformly restricted diamond distance, in which the optimization is taken over bipartite input states whose marginal on the system of interest is maximally mixed. This leads to the following results (abbreviated from Theorem \ref{th:learning_apx_sparse} in Section \ref{sec:l1_bounded}, and Theorem \ref{th:lower_l1} in Section \ref{sec:lb}):
\begin{result}
    There exists a polynomial time quantum algorithm that learns unitaries with Pauli $\ell_1$-norm uniformly bounded by $L_1$ under diamond distance using $\Omega(2^{n/2})$ queries, but  under  the uniformly restricted diamond distance using 
    $ \tilde{O}( \frac{L_1^8}{\epsilon^{16}})$ queries.
\end{result}
  
Furthermore, in Section \ref{sec:examples} we present natural looking examples of unitaries that are strictly nearly sparse or have small Pauli $\ell_1$-norm.

Finally, in Table~\ref{tab:papers}, we summarize  our results in position within the broader literature on operator learning. Existing Hamiltonian learning algorithms \citep{abbas2025nearly,Arunachalam2025} achieve favorable sample complexity by exploiting access to short-time evolutions, but  rely on the ability to query the underlying generator. In contrast, our work focuses on learning a fixed unitary operator $U$ directly, without assuming time-control access or restricting $U$ to a particular subgroup of unitaries \citep{Grewal2025}.

\begin{table}[h!]
    \centering
    \small
    \renewcommand{\arraystretch}{1.3} 
    \begin{tabular}{c|c|c}
    
    \textbf{Work} & \textbf{Target Object} & \textbf{Query Complexity} \\
    \hline
    \citep{abbas2025nearly} & $s$-sparse Hamiltonian ($H$) & $\tilde{O}(s^2/\epsilon)$  variable time  query to $e^{-iHt}$ \\
    \hline
    \citep{Arunachalam2025} & $s$-sparse Hamiltonian ($H$) & $\tilde{O}(s^2/\epsilon^4)$ fixed time  query to $e^{-itH}$ \\
    \hline
    \citep{Grewal2025} & \makecell{Unitary  with $2^{k}$ Pauli\\ subspace dimensionality}  & ${O}(2^k/\epsilon)$  query to $U$ \\
    \hline
    Theorem \ref{th:learning_n_sparse} & nearly $(s, \epsilon)$-sparse unitary   & $\tilde{O}(s^6/\epsilon^4)$ query to $U$\\
    \hline
    Theorem \ref{th:learning_apx_sparse} & $\ell_1$ bounded unitary by $L_1$ & $\tilde{O}( L_1^8/\epsilon^{16})$ query to $U$ \\
    
    \end{tabular}
    \caption{Comparison of query complexities for various  learning problems.}
    \label{tab:papers}
\end{table}

\paragraph{Organization}
The remainder of this paper is organized as follows. Section \ref{sec:related works} discusses related work in more detail.
In Section \ref{sec:pre}, we establish the necessary definitions, notation and preliminaries regarding Pauli decompositions and the diamond norm.  Section \ref{sec:Pauli est} presents our algorithm for Pauli coefficient estimation along with its analysis. Section \ref{sec:Ulearning} presents our learning algorithms and derives the sample complexity bounds for learning nearly-sparse unitaries. Section \ref{sec:l1_bounded} studies the learnability of unitaries with bounded Pauli $\ell_1$-norm under the restricted diamond distance, along with Section \ref{sec:lb} lower bounds for learning unitaries under the standard diamond distance.  Section \ref{sec:examples} discusses special cases and examples of our results. Finally, Section \ref{sec:conclusion} concludes with a discussion of open questions and future directions.

\section{Related Work}\label{sec:related works}

While learning a general unitary requires $O(d^2)$ queries \citep{gutoski2014process,Haah2016}, recent work by Haah et al. \citep{Haah2023} refined these bounds specifically for the diamond distance metric, establishing the query optimality for distinguishing unitary channels.
Because of this barrier, many recent works have employed heuristic methods, including gradient descent \citep{kiani2020learning}, classical shadows \citep{levy2024classical}, and quantum neural networks \citep{beer2020training}.

While these methods are useful, they often lack guarantees of efficiency. Thus, a valid question is: \emph{which classes of unitaries are efficiently learnable?} Zhao et al. studied how to efficiently learn unitaries of bounded gate complexity \citep{Zhao2024}. Similarly, testing and learning quantum $k$-juntas has been a topic of interest due to their real-world application \citep{Chen2022junta,belovs2015}.

The problem of learning sparse Hamiltonians or unitaries in the Pauli basis is closely related to learning sparse Boolean functions. The celebrated KM algorithm \citep{KM1993} or  Goldreich-Levin  \citep{Goldreich1989} have been shown to be effective for learning several classes of Boolean functions including  Disjunctive Normal Form (DNF) expressions, juntas, and decision trees \citep{Jackson1997,Valiant1984,Bshouty95,Kushilevitz96,ABKKPR1998,HellersteinKSS12}. The approach was elaborated and extended by several authors \citep{BshoutyJT04,Feldman07,Kalai2009,Feldman2012,HeidariK2025}.

Several studies have attempted to develop a quantum variant of these algorithms, including Montanaro \citep{Montanaro2010} and Angrisani et al. \citep{angrisani2025learning}. Crucially, Angrisani's method yields a sample complexity that depends on the system dimension $n$, typically to locate the support of the operators.

There is another line of work on learning sparse Hamiltonians. Traditional methods often rely on preparing either an eigenstate or the thermal (Gibbs) state \citep{anshu2020sample}. Recently, Arunachalam et al. used Bell sampling and shadow tomography to introduce a protocol for learning Pauli coefficients of a sparse Hamiltonian with $O(s^2/\varepsilon^4)$ queries \citep{arunachalam2024testing}. Abbas et al. also implemented a method using the isolation technique to learn the coefficients using $\tilde{O}(s \log(1/\varepsilon))$ queries, achieving Heisenberg-limited scaling \citep{abbas2025nearly}.

PAC learning of quantum processes has also been studied in recent literature \citep{Heidari2024,HeidariQuantum2021,Arunachalam2018,Arunachalam2017}. These works focus on learning quantum channels under the PAC framework, where the learner aims to approximate an unknown quantum channel based on input-output examples drawn from a specific distribution. Their results provide sample complexity bounds for learning quantum channels with respect to various distance measures, such as the diamond norm and trace norm. However, these works do not specifically address the efficient learnability of sparse unitaries in the Pauli basis, which is the focus of our study. 

In the context of unitary learning, the learnability depends heavily on the available access model and the structural assumptions. 
If the unitary is generated by a sparse Hamiltonian and one has access to the time-evolution operator at variable times, \citet{abbas2025nearly} provides a Hamiltonian learning algorithm with optimal linear dependence on sparsity.
In the standard black-box setting (fixed $U$), \citet{arunachalam2024testing} show that unitaries generated by sparse Hamiltonians can be learned up to Pauli $\ell_\infty$-norm.
Additionally, unitaries whose Pauli spectrum is supported on a small subgroup are learnable with complexity $\tilde{O}(2^k/\epsilon)$.

Our work targets a broader class: general unitaries that are sparse (or nearly sparse) in the Pauli basis, without assuming a underlying sparse Hamiltonian structure or time-control access.
While \citet{abbas2025nearly} achieve better complexity, they require stronger access (variable time evolution). We introduce an efficient learning algorithm $\tilde{O}(s^6/\epsilon^4)$ using only standard access to the fixed unitary $U$.

\section{Preliminaries}\label{sec:pre}

\noindent\textbf{Notation.}
We consider an $n$-qubit quantum system with Hilbert space dimension $N = 2^n$.
For any positive integer $d$, we use the shorthand $[d] := \{1,2,\ldots,d\}$.

\subsection{Learnability Formulation}
We consider the problem of learning an unknown $n$-qubit unitary operator $U$ from \emph{coherent} querying. 
That is the learner can make can prepare entangled input states $\ket{\phi}$, possibly in a larger space of $n+m$ qubits; apply the unknown unitary $U$ on a $n$-qubit subsystem of choosing; and perform joint measurements on the output states;
$(I\tensor U)\ket{\phi}$. 

The goal is to output a hypothesis unitary $\hat{U}$ that approximates $U$ well with respect to a specified distance metric.
\begin{definition}
[Learning unitaries]
Let $\ClC$ be a concept class of $n$-qubit unitaries. We say that $\ClC$ is  learnable with sample complexity $T$ and error $\varepsilon$ with respect to a distance metric $d(\cdot,\cdot)$ if there exists a learning protocol that makes $T$ queries to the unknown unitary $U\in \ClC$ and outputs a hypothesis operator $\hat{U}$ such that with probability at least $2/3$,
\[
	d(U,\hat{U}) \leq \varepsilon.
\]
 The query complexity is the minimum $T$ for which such a learning protocol exists. The learning is efficient if the learning protocol runs in time polynomial in $n$, $T$, and other relevant parameters.
\end{definition}

In this work, we consider learning two classes of unitaries: (i) nearly $s$-sparse unitaries in the Pauli basis (see Definitions \ref{def:sparse} and \ref{def:nearly sparse}); and (ii) Unitaries with bounded Pauli-1 norm. Moreover, we the  distance metric in this model is the diamond norm and a relaxed version of that.

\subsection{Pauli Decomposition}
In this work, we focus on learning $n$-qubit unitaries with representation in the Pauli basis. 
The Pauli operators together with the identity  are denoted as $\{\sigma^0, \sigma^1, \sigma^2, \sigma^3\}$ with
\begin{align*}
	\sigma^0 = I= \begin{pmatrix}
		              1 & 0 \\ 0 & 1
	              \end{pmatrix}
	\qquad
	\sigma^1 =X = \begin{pmatrix}
		              0 & 1 \\ 1 & 0
	              \end{pmatrix},
	\quad  \sigma^2 = Y=  \begin{pmatrix}
		                      0 & -i \\ i & 0
	                      \end{pmatrix},
	\quad  \sigma^3 =Z =  \begin{pmatrix}
		                      1 & 0 \\ 0 & -1
	                      \end{pmatrix}.
\end{align*}
For an index vector $\bfs = (s_1,\ldots,s_n) \in \{0,1,2,3\}^n$, define the $n$-qubit Pauli operator
$$
\sigma^{\bfs} := \sigma^{s_1} \otimes \sigma^{s_2} \otimes \cdots \otimes \sigma^{s_n}.
$$

Then the set $\qps, s\in \paulisetn$ contains all $4^n$ $n$-qubit Pauli operators and they form an orthogonal basis for the space of $n$-qubit operators with respect to the Hilbert-Schmidt inner product.

\begin{fact}
	Any linear operator $A$ acting on $n$ qubits admits a unique Pauli decomposition 
	$$A = \sum_{\bfs \in \paulisetn} \alpha_{\bfs} ~\sigma^{\bfs},$$
	where the coefficients are given by $\alpha_\bfs = \frac{1}{2^n}\Tr\big(A\qps\big).$ 
    
\end{fact}

We also define norms on the Pauli coefficient vector of an operator $A$.
Recalling that $A = \sum_\bfs \alpha_\bfs \sigma^\bfs$, we define
\begin{align*}
\norm{A}_{1,P} &:= \sum_\bfs |\alpha_\bfs|, \\
\norm{A}_{2,P} &:= \sqrt{\sum_\bfs |\alpha_\bfs|^2}, \\
\norm{A}_{\infty,P} &:= \max_\bfs |\alpha_\bfs|.
\end{align*}

\begin{definition}
[Pauli sparsity]\label{def:sparse}
An $n$-qubit unitary $U$ is called $s$-sparse in the Pauli basis if it has at most $s$ non-zero Pauli coefficients in its Pauli decomposition: $U = \sum_{j=1}^s \alpha_{\bfs_j} \sigma^{\bfs_j}$.

\end{definition}

Our focus is on the broader class of nearly sparse unitaries:
\begin{definition}
[Nearly sparse unitaries]\label{def:nearly sparse}
An $n$-qubit unitary $U$ is called nearly $(s,\epsilon)$-sparse in the Pauli basis if there is a set $\CS$ of size at most $s$, such that  $\sum_{\bfs\notin \CS} |\alpha_\bfs| \leq \epsilon$.

\end{definition}

\subsection{Other relevant distance metrics and norms}
We use several operators and channel norms throughout this work.
For any linear operator $A$, the trace norm is defined as
$\norm{A}_1 := \Tr(\sqrt{A^\dagger A})$,
the Frobenius norm is
$\norm{A}_2 := \sqrt{\Tr(A^\dagger A)}$,
and the operator norm is
$$
\norm{A}_{op} := \max_{\ket{\psi} : \norm{\ket{\psi}}_2 = 1} \norm{A\ket{\psi}}_2.
$$

The Frobenius norm and the Pauli $2$-norm are related by
\begin{equation}
\norm{A}_2 = \sqrt{N} \norm{A}_{2,P},
\label{eq:eqof2l2}
\end{equation}
where $N = 2^n$ is the Hilbert space dimension.

Below we define slightly generalizations of diamond and phase aligned operator distance.

\begin{definition}[diamond distance]\label{def:diamond}
	For a pair of quantum mappings $\CE$ and $\CF$, the diamond distance is defined as
$$
d_\diamond(\CE, \CF) := \max_\rho \norm{(\II \tensor (\CE - \CF))(\rho)}_1,
$$
where the maximization is over all density operators $\rho$ acting on a space of
dimension $N^2$.
\end{definition}

\begin{definition}[Phase-aligned operator distance]
For operators $U,V$, the \emph{phase-aligned operator norm distance} is defined as
\begin{equation}
d_{optphase}(U, V) :=
\min_{\theta \in [0,2\pi)}
\norm{U - e^{i\theta} V}_{\mathrm{op}} .
\end{equation}
\end{definition}

Next, we prove an inequality connecting the above two distances.

\begin{lemma}[Generalized Distance Bound]
    \label{lm:gen_l1_to_diamond}

    Let $U$ and $V$ be arbitrary linear operators acting on $n$ qubits. Let $\mathcal{U}(\rho) = U\rho U^\dagger$ and $\mathcal{V}(\rho) = V\rho V^\dagger$. Then the diamond distance between the induced maps is bounded by:
    $$ d_\diamond(\mathcal{U}, \mathcal{V}) \le (\|U\|_{op} + \|V\|_{op}) d_{optphase}(U, V) $$
    Moreover, 
    $$  d_{optphase}(U, V)\leq \min_{\phi \in [0, 2\pi]}  \|U - e^{i\phi}V\|_{1,P} $$
    
\end{lemma}
\begin{proof}
    Let $\rho_{RA}$ be an arbitrary state on the joint system. 
    Let $\phi$ be an arbitrary phase and define $E = U - e^{i\phi}V$.   We can rewrite the difference acting on $\rho_{RA}$ by adding and subtracting $(\mathbb{I} \otimes U)\rho_{RA}(\mathbb{I} \otimes e^{-i\phi}V^\dagger)$. Therefore, as $\mathcal{U}$ and $\mathcal{V}$ are independent of the global phases of $U$ and $V$ we have that
    \begin{align*}
        (\mathbb{I} \otimes \mathcal{U})(\rho_{RA}) - (\mathbb{I} \otimes \mathcal{V})(\rho_{RA}) &=  (\mathbb{I} \otimes U)\rho_{RA} (\mathbb{I} \otimes U^\dagger) - (\mathbb{I} \otimes e^{i\phi}V)\rho_{RA} (\mathbb{I} \otimes e^{-i\phi}V^\dagger)\\
        &= (\mathbb{I} \otimes U)\rho_{RA}(\mathbb{I} \otimes E^\dagger) + (\mathbb{I} \otimes E)\rho_{RA}(\mathbb{I} \otimes e^{-i\phi}V^\dagger).
    \end{align*}
    Applying the triangle inequality together with the inequalities  
    
    (a) $\|XYZ\|_1 \le \|X\|_{op}\|Y\|_1\|Z\|_{op}$; (b) $\|\rho_{RA}\|_1 = 1$; and (c) $\|A\otimes B\|_{op} = \|A\|_{op}\|B\|_{op}$:
    \begin{align*}
        \| (\mathbb{I} \otimes \mathcal{U} - \mathbb{I} \otimes \mathcal{V})(\rho_{RA}) \|_1 
        &\le \|U\|_{op} \cdot 1 \cdot \|E^\dagger\|_{op} + \|E\|_{op} \cdot 1 \cdot \|V^\dagger\|_{op} \\
        &= \|U\|_{op}\|E\|_{op} + \|E\|_{op}\|V\|_{op}
    \end{align*}

    Taking the supremum over all states $\rho_{RA}$ and the minimum over all phases $\phi$ establishes the first inequality.
    For the second part of the lemma, let the Pauli decomposition $U=\sum_\bfs \alpha_\bfs \qps$ and $V=\sum_\bfs \beta_\bfs \qps$. Then, for any $\phi$, we can write 
\begin{align*}
    d_{optphase}(U, V)\leq \norm{U - e^{i\phi}V}_{op} & = \norm{\sum_{\bfs} ( \alpha_\bfs - e^{i\phi}\beta_\bfs ) \sigma_{\bfs}}_{op}\\
    &\leq \sum_{\bfs} \abs{\alpha_\bfs - e^{i\phi}\beta_\bfs} \norm{\sigma_{\bfs}}_{op} = \norm{U - e^{i\phi}V}_{1,P},
\end{align*}
where we used  the fact
that each Pauli operator has operator norm equal to one.

\end{proof}
    If both $U$ and $V$ are unitary, the lemma reduces to the upper bound in \citep[Proposition 1.6]{haah2023query}, yielding:
    $$d_\diamond(\mathcal{U}, \mathcal{V}) \le 2d_{optphase}(U, V) .$$
\begin{corollary}
    \label{cor:gen_l1_to_dim}
    When $U$ is unitary, 
    $$d_\diamond(\mathcal{U}, \mathcal{V}) \le \min_{\phi \in [0, 2\pi]} \left( 2 \|U - e^{i\phi}V\|_{1,P} + \|U - e^{i\phi}V\|_{1,P}^2 \right).$$

\end{corollary}
\begin{proof}
    The argument follows from the above lemma and the triangle inequality: $\|V\|_{op} = \|U - E\|_{op} \le \|U\|_{op} + \|E\|_{op}$.
\end{proof}

\subsection{Shadow Tomography}

The task of estimating properties of an unknown quantum state has a long history. Shadow tomography, introduced by Huang et al. \citep{Huang2020}, provides a protocol for estimating the expectation values of $M$ observables using a number of samples that scales logarithmically with $M$. This stands in contrast to full quantum state tomography, whose sample complexity scales exponentially with the system size.
In this work, we employ shadow tomography based on random Clifford measurements.

\paragraph{Clifford group.}
The $n$-qubit Clifford group $\mathcal{C}_n$ is defined as the normalizer of the
$n$-qubit Pauli group. That is, for any $U \in \mathcal{C}_n$ and any Pauli operator
$\sigma^{\bfs}$, we have
$$
U \sigma^{\bfs} U^\dagger = \pm \sigma^{\bfs'}
$$
for some Pauli operator $\sigma^{\bfs'}$.
The Clifford group is generated by the Hadamard gate, the phase gate, and the CNOT gate.

 A fundamental property of Clifford circuits is that they admit efficient classical simulation: by the Gottesman--Knill theorem, any quantum computation consisting solely of Clifford unitaries, stabilizer state preparations, and Pauli measurements can be simulated in time polynomial in $n$ on a classical computer.

\paragraph{Shadow tomography with Clifford measurements.}

To construct a classical shadow, as in Algorithm \ref{alg:shadow}, we apply a random unitary $U$ drawn uniformly from $\mathcal{C}_n$ to the state $\rho$ and measure in the computational basis to obtain a bitstring $b \in \{0,1\}^n$. The classical snapshot is constructed as:
\begin{equation}
\hat{\rho} = (2^n + 1) U^\dagger \ketbra{b} U - \mathbb{I}.
\end{equation}
The random matrix $\hat{\rho}$ is an unbiased estimator of $\rho$, meaning $\mathbb{E}[\hat{\rho}] = \rho$. The efficiency of this protocol depends on the variance of the estimator, which is governed by the shadow norm $\norm{O}_{\text{shadow}}$. For random Clifford measurements, $\norm{O}_{\text{shadow}}^2$ is closely related to the Hilbert-Schmidt norm $\Tr(O^2
)$. This is particularly advantageous for global or low-rank observables—such as the Pauli string projectors used in our learning algorithm—where this norm is small, leading to highly efficient estimation. In the general case, to predict $M$ observables $\{O_i\}$ within error $\varepsilon$, the required sample complexity is $T = O(\max_i \norm{O_i}_{\text{shadow}}^2 \log(M) / \varepsilon^2)$. Since the observables used in our algorithm are constructed from low-rank Bell basis states (rank-2 operators with bounded spectral norm), their shadow norms are bounded by a small constant, leading to the following specific guarantee:

\begin{theorem}[cf. Theorem 2 in \citep{Huang2020}]
\label{thm:CST}
Let $\{O_i\}_{i=1}^M$ be observables   $\norm{O_i}_{op} \leq 1$ for all
$i \in [M]$, and assume each $O_i$ admits a classical description that can be
simulated in polynomial time.
For any $\varepsilon, \delta > 0$, there exists a procedure that, given
$$
T = O\!\left(\frac{\log(M/\delta)}{\varepsilon^2}\right)
$$
independent snapshots of $\rho$ obtained via random Clifford measurements, outputs
estimates $\hat{o}_i$ such that with probability at least $1 - \delta$,
\begin{equation}
|\hat{o}_i - \Tr(O_i \rho)| \leq \varepsilon
\quad \text{for all } i \in [M].
\end{equation}
\end{theorem}

\begin{algorithm2e}[ht]
\DontPrintSemicolon
\LinesNumbered
\caption{Shadow Tomography}
\label{alg:shadow}

\KwIn{$m$ copy of an state $\ket{\psi}$,  observable set $\CO=\{O_j: j\in k\}$.}
\KwOut{Estimate of expectation value  $\hat{o}_i$ for ${i \in [k]}$.}

\SetKwFunction{CST}{CST}

\SetKwProg{Fn}{Function}{:}{end}

\Fn{\CST{$\CO, m, \ket{\psi}$}}{

\For{$j\in [m]$}{
    Sample and construct a Clifford unitary $V_j$ randomly.\;

    Apply $V_j$ on the $j$th copy of $\ket{\psi}$. \;

    Measure in the computational basis and record the output to be $\bfb_j\in \{0,1\}^n$.\;
}
\For{$O_i \in \CO$}{
    Calculate $\hat{o}_{j} :=(2^n+1)\expval{V_j^\dagger O_i V_j}{\bfb_j}$ for all $j\in [m]$\;
 
    Apply the Median-of-Means estimator on $\hat{o}_{j}, j\in [m]$ and record the results to $ \hat{o}_i$.\;

}
\Return $\set{\hat{o}_i}_{i=1}^{k}$\;
}
\end{algorithm2e}

\subsection{Bell sampling and Choi-Jamiołkowski Isomorphism}
Bell sampling refers to the procedure of performing a joint measurement of two qubits in the Bell basis, rather than in the computational basis. The Bell basis consists of the four maximally entangled two-qubit states
\begin{equation}
\ket{\Phi^\pm} = \frac{1}{\sqrt{2}}(\ket{00} \pm \ket{11}), \qquad
\ket{\Psi^\pm} = \frac{1}{\sqrt{2}}(\ket{01} \pm \ket{10}),
\end{equation}
which form an orthonormal basis of the two-qubit Hilbert space.   
Bell measurements are particularly useful for estimating expectation values of Pauli operators used in shadow tomography. Since tensors of Pauli words $\qps\tensor\qps$ commute, bell measurements can be used to simultaneously measure all $n$-qubit Pauli observables on two copies of an $n$-qubit state \citep{Huang2021}. 

The Choi-Jamiołkowski isomorphism establishes a one-to-one correspondence between quantum channels and quantum states.
In particular, for a unitary operator $U$ acting on an $n$-qubit system, the
associated Choi state is a pure state on $2n$ qubits defined as
\begin{equation*}
\ket{J(U)} \defeq (U \otimes \II_N)
\left(
\frac{1}{\sqrt{N}} \sum_{i=0}^{N-1} \ket{i} \ket{i}
\right),
\end{equation*}
where $N = 2^n$.

The Choi state can be prepared by generating $n$ EPR pairs and applying the unitary $U$ to one half of each pair. Particularly relevant to our work is the relationship between the Pauli coefficients of $U$ and the expectation values of Pauli observables on the Choi state $\ket{J(U)}$.   

Define the (generalized) Bell states by
\begin{equation}\label{eq:bell basis of P}
\ket{\phs} = (\sigma^\bfs \otimes \II) \ket{EPR}^{\otimes n}, \qquad \bfs \in \paulisetn,
\end{equation}
where $\ket{EPR} = \frac{1}{2} (\ket{00} + \ket{11})$ is the maximally entangled state.

The family $\{\ket{\phs}\}$ forms an orthonormal basis of the $N^2$ dimensional Hilbert space of bipartite systems. Measuring  in this basis is referred to as the $n$-qubit Bell measurement. A modern reference that explicitly defines Bell basis states  \citep{Huang2021}. 

\begin{lemma}
[cf. Lemma 3 of \citep{Huang2021}]\label{lem:bell sampling}
 Measuring the Choi state $\ket{\Phi_U}$ in the Bell basis, samples the distribution given by the squared magnitudes of the Pauli coefficients $|\qas|^2$ of $U$. That is 
 \[
	P_\bfs = |\braket{\phs}{J(U)}|^2 = |\qas|^2
 \]

\end{lemma}

In particular, if $U$ is $s$-sparse in the Pauli basis, Bell sampling outputs a nonzero Pauli operator with probability one and identifies its label with probability proportional to the squared magnitude of its coefficient.

\subsection{Block Encoding}
\ac{LCU} allows for the implementation of an operator $A$ by expressing it as a weighted sum of efficiently implementable unitaries, $A = \sum_i \alpha_i U_i$. This approach is particularly effective for operators that can be well-approximated by sparse Pauli decompositions, facilitating the translation of mathematical structures into gate-based quantum circuits \citep{childs2012hamiltonian}.

Specifically, the success probability of this implementation is $1/\alpha^2$, where $\alpha = \sum_s |\alpha_s|$ is the $L_1$-norm of the learned Pauli coefficients. To achieve a near-deterministic implementation, we employ Oblivious Amplitude Amplification (OAA). Standard amplitude amplification requires knowledge of the input state to construct the appropriate reflection; however, OAA is "oblivious" in that it succeeds for any input state $|\psi\rangle$, provided the operator being implemented is (close to) a unitary. In the context of our sparse unitary learning algorithm, if the scaling factor is chosen such that $\alpha = 2$, a single step of OAA

perfectly implements the unitary without ancilla measurements. For general $\alpha$, or when the learned coefficients contain approximation errors, we utilize the Robust Oblivious Amplitude Amplification technique. This involves a sequence of fixed-point rotations that boost the fidelity of the implementation to $1-\delta$ with a query overhead that scales as $O(\alpha \log(1/\delta))$. This ensures that the reconstruction of the learned $s$-sparse unitary remains efficient and compatible with larger quantum algorithms.

\begin{theorem}[c.f. Theorem 2.5 \citep{Kothari2014}]\label{thm:approx LCU}
Let $\tilde{V}=\sum_i a_i U_i$ is a linear combination of unitary matrices $U_i$ with $a_i>0$ for all $i$.

Let $A$ be a unitary that maps
\begin{equation*}
\ket{0^m} \;\longmapsto\; \frac{1}{\sqrt{a}} \sum_i \sqrt{a_i}\ket{i},
\qquad
\text{where }
a := \norm{\vec{a}}_1 = \sum_i a_i .
\end{equation*}
Then, if $\tilde{V}$ is $\delta$-close to some unitary in spectral norm, and $\tfrac{1+\delta}{a}\leq 1$,  there exists a quantum algorithm that implements the map $\tilde{V}$ with error $O(a\sqrt{\delta})$ and makes
$O(a)$ uses of $A$, the select unitary
\begin{equation*}
U := \sum_i \ket{i}\!\bra{i} \otimes U_i,
\end{equation*}
and their inverses.
\end{theorem}

 \section{Finding And Estimating Large Pauli Coefficients}\label{sec:Pauli est}
In this section, we present an algorithm to find and estimate large Pauli coefficients of an unknown unitary $U$. It consists of two subroutines: (1) finding the indices of large Pauli coefficients, and (2) estimating the values of those coefficients.

The goal in this section is to estimate the Pauli coefficients using the fewest possible queries of an unknown unitary. First, we discuss how to calculate $\alpha_{\bft}$, then implement an algorithm to estimate other Pauli coefficients.

Motivated by classical sparse recovery algorithms such as the KM algorithm \citep{KM1993}, the building block of the learning algorithm is estimating large Pauli coefficients of an unknown unitary $U$. Since $U$ is not necessarily a small perturbation of identity, we cannot directly use Hamiltonian learning algorithms \citep{arunachalam2024testing,abbas2025nearly} that assume access to $e^{-iHt}$ for small $t$ to linearize the problem. In addition, since $U$ is not necessarily Hermitian, we cannot directly apply shadow tomography techniques to estimate Pauli observables \citep{Chen2024}. Indeed, without access to controlled-$U$, estimating Pauli coefficients is impossible due to the global phase ambiguity in $U$. 

Therefore, we can hope to estimate the large coefficients up to a global phase. We adapt the Pauli shadow method of \citep{arunachalam2024testing} to first find a list of large Pauli coefficients and estimate the magnitude of the biggest one first.

\subsection{Finding index of large coefficients}
To find the indices of large Pauli coefficients, we use Bell sampling of the Choi state of $U$. According to Lemma \ref{lem:bell sampling}, measuring the Choi state $\ket{J(U)}$ in the Bell basis produces an outcome $\bfs\in \pauliset^n$ with probability $|\alpha_\bfs|^2$. Therefore, by preparing multiple copies of $\ket{J(U)}$ and measuring them in the Bell basis, we can collect samples from the distribution $\{|\alpha_\bfs|^2\}$. Specifically, if we take enough samples, we  ensure that with high probability, all indices corresponding to Pauli coefficients with magnitude above a certain threshold $\theta$ are included in our sample set. Denote this set by $\CX$. This procedure is summarized as Algorithm \ref{alg:find_coef}.

\begin{algorithm2e}[ht]
\DontPrintSemicolon
\LinesNumbered

\caption{Finding Large Coefficients}
\label{alg:find_coef}

\KwIn{Query access to unitary $U$; threshold $\theta\in(0,1)$; $\epsilon, \delta\in(0,1)$.}
\KwOut{Index Set $\CX$ and estimate of the largest Pauli coefficient $\hat{\alpha}_{max}$}

\BlankLine

\textbf{Set} 
$m_1 \gets O\!\left(\dfrac{\log(1/\theta)+\log(1/\delta)}{\theta^2}\right)$,
\quad
and initialize $\CX=\emptyset$ \;

Prepare $m_1$ copies of the Choi state $|J(U)\rangle = (U\otimes \II_{2^n})\,\ket{EPR}^{\tensor n}$. \;

Measure each copy in the Bell basis and add the outcome $\bfs\in \pauliset^n$ to $\CX$. \;

\Return $\CX$.\;

\end{algorithm2e}

The next result shows that $\CX$ contains the index of all large Pauli coefficients of $U$:
\begin{lemma}
    With probability at least $1-\delta$, the set $\CX$ generated in Algorithm \ref{alg:find_coef} contains the set  $\{\bfs: \abs{\alpha_\bfs} \geq \theta\}$ and $|\CX| \leq  \frac{\ln{(1/(\theta^2\delta))}}{\theta^2}$. 
    \label{lm:coupon}
\end{lemma}
\begin{proof}

    We define the target set as $S_\theta = \{\bfs : |\alpha_\bfs| \ge \theta\}$. Suppose $S_\theta$ is not empty. 
    
    This is related to the non-uniform variant of the famous Coupon Collection problem. 
    The probability of observing any specific target index $\bfs \in S_\theta$ in a single shot is $p(\bfs) = |\alpha_\bfs|^2 \ge \theta^2$,

    and the probability that $\bfs$ is never observed in $m_1$ independent draws from $p$ is bounded by $e^{-\theta^2 m_1}$,
    as 
    $(1-p_\bfs)^{m_1} \leq e^{-\theta^2 m_1}.$
    
    By the union bound, the probability of missing an index in $S_\theta$ after $m_1$ samples is
    $$Pr[\exists \bfs \in S_\theta \text{ missing after $m_1$ samples}]\leq \sum_{\bfs\in S_\theta} (1-p(\bfs))^{m_1} \leq \abs{S_\theta}e^{-\theta^2 m_1}.$$
    
    Hence, to have this bounded by $\delta,$ we need 
    $m_1 \geq \frac{\ln(\abs{S_\theta}/\delta)}{\theta^2} $.
    Since $U$ is unitary, by Parseval's identity
    , $\sum_\bfs |\alpha_\bfs|^2 =1$, that implies $1 \ge \sum_{\bfs \in S_\delta} \abs{\alpha_\bfs}^2 \ge \theta^2 \abs{S_\theta}$.
    Thus, the number of significant coefficients is bounded by
    $\abs{S_\theta} \leq 1/\theta^2$. Hence, the query bound $m_1\geq \frac{\ln{(1/(\theta^2\delta))}}{\theta^2}$ is obtained. The second statement on $|\CX|$ is immediate as $|\CX|\leq m_1$.

\end{proof}

\subsection{Estimating large coefficients up to a global phase}
Given the list $\CX$, the next step is to estimate the Pauli coefficients in $\CX$. Since the coefficients are complex numbers in general we need a special care for the estimation. 

We first estimate the $|\alpha_\bfs|^2, \bfs\in \CX$ with error $\varepsilon$ and pick the largest, indexed by $\bft$. Then we declare $\hat{\alpha}_\bft  = |\alpha_\bft|$ as the estimate of the Pauli coefficient corresponding to $\bft$.  That is we ignore the phase of $\alpha_\bft$. This can be done by estimating the expectation value of  the observable 
$$M_{\bfs} \coloneqq \ketbra{\phs}{\phs}, $$
for $\bfs\in \CX$, where $\ket{\phs}$ is the Bell basis state corresponding to $\bfs$.
 Measuring $M_{\bfs}$ on the Choi state $\ket{J(U)}$ gives
 $\expval{M_{\bfs}}{J(U)} = |\alpha_\bfs|^2.$

Next, with $\hat{\alpha}_\bft$ obtained,   we estimate the other coefficients $\alpha_\bfs, \bfs\in \CX$. Since we do not know the phase of $\alpha_\bft$, we can only estimate other coefficients up to the same global phase.  More precisely, we estimate the ratio $\alpha_\bfs/\alpha_\bft$ for all $\bfs\in \CX$.
This can be done via the following observables done in the Bell basis:
$$R_{\bft,\bfs} \coloneqq \frac{1}{2}(\ketbra{\phi_{\bft}}{\phs} + \ketbra{\phs}{\phi_\bft})$$
$$I_{\bft,\bfs} \coloneqq \frac{1}{2}(-i\ketbra{\phi_{\bft}}{\phs} + i\ketbra{\phs}{\phi_\bft})$$
for $\bft,\bfs \in \set{0, 1, 2, 3}^n$.

We measure several copies of the Choi state $\ket{J(U)}$ of $U$ to estimate expectation values of all $R_{t,\bfs}, I_{t,\bfs}$.

To reduce the sample complexity, we use shadow tomography with Clifford measurements. Note that  $R_{t,\bfs}, I_{t,\bfs}$ belong to the clifford group, and hence can be simulated Classically. As a result, the shadow tomography can reduce the sample complexity to scale logarithmically with $|\CX|$ with polynomial time. This is summarized in Algorithm \ref{alg:estimate-unitary}.

\begin{algorithm2e}[ht]
\DontPrintSemicolon
\LinesNumbered
\caption{Estimating Pauli Coefficients}
\label{alg:estimate-unitary}

\KwIn{Query access to unitary $U$; index set $\CX$, accuracy $\varepsilon\in(0,1)$, failure probability $\delta\in(0,1)$.}
\KwOut{Estimate of the Pauli coefficients  $\hat{\alpha}_\bfs$ for ${\bfs\in \CX}$.}

\BlankLine

Set 
$m_2 \gets O \left(\frac{\log ({|X|} / {\delta})}{\varepsilon^2} \right)$. \;

Prepare 2 set of $m_2$ copies of the Choi state $|J(U)\rangle = (U\otimes \II_{2^n})\,\ket{EPR}^{\tensor n}$. \;

 $\tilde{M}_\bfs \gets$ \CST{$\set{M_\bfs : \bfs\in \CX}, m_2 , \ket{J(U)}$}\;
 
$\bft\gets \argmax \tilde{M}_\bfs$ and $\hat{\alpha}_\bft\gets \tilde{M}_\bft$\;

$\tilde{R}_\bfs, \tilde{I}_\bfs \gets$ \CST{$\set{R_{\bft,\bfs}, I_{\bft,\bfs} : \bfs\in \CX}, m_2 , \ket{J(U)}$}\;

$\hat{\alpha}_\bfs \gets \frac{\tilde{R}_\bfs + i \tilde{I}_\bfs}{\hat{\alpha}_\bft}$ for all $\bfs\in \CX$.\;

\Return $\{\hat{\alpha}_\bfs: \bfs\in \CX\}$\;
\end{algorithm2e}

\begin{lemma}

    Algorithm \ref{alg:estimate-unitary} with $\CX, \delta, \varepsilon$ as the inputs, with probability at least $1-\delta$, estimates all the Pauli coefficients in $\CX$  up to an additive  error $\frac{3\sqrt{\varepsilon}}{|{\alpha}_{\bft}|-\sqrt{\varepsilon}} $ and a global phase $\phi\in [0,\pi]$; that is, for all $\bfs\in \CX$,
    $$|\hat{\alpha}_\bfs - e^{-i\phi}\alpha_\bfs| \leq \frac{3\sqrt{\varepsilon}}{|{\alpha}_{\bft}|-\sqrt{\varepsilon}} .$$
        \label{lm:pauli_error}
\end{lemma}

\begin{proof}
    From Theorem \ref{thm:CST}, the estimates $\tilde{M}_\bfs$,  $\tilde{R}_{\bft,\bfs}$ and $\tilde{I}_{\bft,\bfs}$ in Algorithm \ref{alg:estimate-unitary}  are $\varepsilon$-accurate; that is, they satisfy:
$$\abs{\Tr[M_{\bfs}\ketbra{J(U)}]- \tilde{M}_{\bfs}} \leq \varepsilon$$
    and
    $$\abs{\Tr[R_{\bft,\bfs}\ketbra{J(U)}]- \tilde{R}_{\bft,\bfs}} \leq \varepsilon, \quad \abs{\Tr[I_{\bft,\bfs}\ketbra{J(U)}]- \tilde{I}_{\bft,\bfs}} \leq \varepsilon$$
    for every $\bfs\in \CX$ with probability $\geq 1- \delta$. 
    
    Note that the Choi matrix is:
    $$ J(U) = \ketbra{J(U)}{J(U)} = \sum_{\bfu, \bfv} \alpha_\bfu \alpha_{\bfv}^* \ketbra{\phi_\bfu}{\phi_{\bfv}}. $$
    By orthonormality ($\braket{\phi_\bfu}{\phi_\bfv} = \delta_{\bfu\bfv}$), $\Tr[\ketbra{\phi_\bfs}{\phi_\bft} J(U)]$ isolates the coefficient $\alpha_\bft \alpha_\bfs^*$:
    $$ \Tr\left[ \ketbra{\phi_\bfs}{\phi_\bft} \sum_{\bfu,\bfv} \alpha_\bfu \alpha_\bfv^* \ketbra{\phi_\bfu}{\phi_\bfv} \right] = \sum_{\bfu,\bfv} \alpha_\bfu \alpha_\bfv^* \braket{\phi_\bft}{\phi_\bfu} \braket{\phi_\bfv}{\phi_\bfs} = \alpha_\bft \alpha_\bfs^*. $$

    Applying this to our observables:
\begin{equation}
    \Tr[M_{\bfs} J(U)] = \Tr[\ketbra{\phi_\bfs}{\phi_\bfs} J(U)] = \alpha_\bfs \alpha_\bfs^* = |\alpha_\bfs|^2,
    \label{eq:M_s}
\end{equation}
\begin{align*}
    \Tr[R_{\bft,\bfs} J(U)] &= \frac{1}{2} \left( \Tr[\ketbra{\phi_\bft}{\phi_\bfs} J(U)] + \Tr[\ketbra{\phi_\bfs}{\phi_\bft} J(U)] \right)  \\
    &= \frac{1}{2}(\alpha_\bfs \alpha_\bft^* + \alpha_\bft \alpha_\bfs^*) = \operatorname{Re}(\alpha_\bfs\alpha^*_{\bft}), \\
    \Tr[I_{\bft,\bfs} J(U)] &= \frac{1}{2i} \left( \ Tr[\ketbra{\phi_\bft}{\phi_\bfs} J(U)] - \Tr[\ketbra{\phi_\bfs}{\phi_\bft} J(U)] \right) \\
    &= \frac{1}{2i}(\alpha_\bfs \alpha_\bft^* - \alpha_\bft \alpha_\bfs^*) = \operatorname{Im}(\alpha_\bfs\alpha^*_{\bft}).
    \label{eq:RI_s}
\end{align*}

    Note that \eqref{eq:M_s} implies that $\abs{\hat{\alpha}_\bft - |\alpha_\bft|} \leq \sqrt{\varepsilon}$. 
    Let the true values of the traces be $R_\bfs = \Tr[R_{\bft,\bfs} J(U)]$ and $I_\bfs = \Tr[I_{\bft,\bfs} J(U)]$. From the above equations, we have that $R_\bfs + i I_\bfs = \alpha_\bfs \alpha_{\bft}^*$.

    Let $\tilde{Z}_\bfs = \tilde{R}_\bfs + i \tilde{I}_\bfs$ be the estimator we obtain in the algorithm. Then write
     \begin{equation*}
        \hat{\alpha}_\bfs = \frac{\tilde{Z}_\bfs}{|\hat{\alpha}_{\bft}|}
    \end{equation*}

    Writing the polar form of $\alpha_{\bft} = e^{i\theta_\bft}|\alpha_{\bft}|$, and using  the identity $$e^{-i\theta_\bft}\alpha_\bfs = e^{-i\theta_\bft}\frac{\alpha_\bfs (\alpha_{\bft} \alpha^*_{\bft})}{|\alpha_{\bft}|^2} = \frac{\alpha_\bfs \alpha^*_{\bft}}{|\alpha_{\bft}|} ,$$

        we can write the error as a difference of fractions:
    \begin{equation*}
        \hat{\alpha}_\bfs - e^{-i\theta_\bft}\alpha_\bfs = \frac{\tilde{Z}_\bfs}{|\hat{\alpha}_{\bft}|} - \frac{\alpha_\bfs \alpha_{\bft}^*}{|\alpha_{\bft}|}.
    \end{equation*}
    To separate the error contributions, we add and subtract the term $\frac{\alpha_\bfs \alpha_{\bft}^*}{|\hat{\alpha}_{\bft}|}$:
    \begin{align*}
        \hat{\alpha}_\bfs - e^{-i\theta_\bft}\alpha_\bfs &= \frac{\tilde{Z}_\bfs}{|\hat{\alpha}_{\bft}|} - \frac{\alpha_\bfs \alpha_{\bft}^*}{|\hat{\alpha}_{\bft}|} + \frac{\alpha_\bfs \alpha_{\bft}^*}{|\hat{\alpha}_{\bft}|} - \frac{\alpha_\bfs \alpha_{\bft}^*}{|\alpha_{\bft}|} \\
        &= \frac{1}{|\hat{\alpha}_{\bft}|} \left( \tilde{Z}_\bfs - \alpha_\bfs \alpha_{\bft}^* \right) + \alpha_\bfs \alpha_{\bft}^* \left( \frac{1}{|\hat{\alpha}_{\bft}|} - \frac{1}{|\alpha_{\bft}|} \right).
    \end{align*}
    Taking the absolute value and applying the triangle inequality:
    \begin{align*}
        \abs{\hat{\alpha}_\bfs - e^{-i\theta_\bft}\alpha_\bfs} &\leq \frac{\abs{Z_\bfs - \alpha_\bfs \alpha_{\bft}^*}}{\abs{\hat{\alpha}_{\bft}}} + \abs{\alpha_\bfs \alpha_{\bft}^*} \frac{||\alpha_{\bft}| - |\hat{\alpha}_{\bft}||}{|\hat{\alpha}_{\bft}| |\alpha_{\bft}|}\\
        &\leq  \frac{2\varepsilon}{\abs{\hat{\alpha}_{\bft}}} + \frac{\abs{\alpha_\bfs}\sqrt{\varepsilon}}{\abs{\hat{\alpha}_{\bft}}}\\
        &=\frac{\sqrt{\varepsilon}}{|\hat{\alpha}_{\bft}|}\left(2\sqrt{\varepsilon} + |\alpha_\bfs|\right)\leq \frac{3\sqrt{\varepsilon}}{|{\alpha}_{\bft}|-\sqrt{\varepsilon}} ,
    \end{align*}
    where by assumption  $||\alpha_{\bft}| - |\hat{\alpha}_{\bft}|| \leq \sqrt{\varepsilon}$, and we used  the fact  that the estimation $\tilde{R}_{t,\bfs}$ and $\tilde{I}_{t,\bfs}$ each have $\varepsilon$ error, hence $Z$ is $ 2\varepsilon$ close to $\alpha_\bfs \alpha_{\bft}^*$.

\end{proof}

This leads to our first main result regarding the efficient estimation of Pauli coefficients:

\begin{theorem}
    Given  $\epsilon, \delta >0$ and $\theta \in (0,1]$, such that $\varepsilon \le \frac{\theta^2}{16}$, there exists an algorithm that, makes 
    
    $$T=O\left( \frac{\log(1/\theta) + \log(1/\delta)}{\epsilon^4} \right)$$
    queries to the target unitary $U$ and,     
    with probability at least $1-\delta$, identifies a set $\CX$ containing the indices of all Pauli coefficients $|\alpha_\bfs|\geq \theta$ and estimates each one, up to a shared global phase factor, with $\frac{\epsilon}{\theta}$ additive error:
    $$|\hat{\alpha}_\bfs - e^{-i\phi}\alpha_\bfs| \leq \frac{\epsilon}{\theta}, \quad \forall \bfs \in \CX,$$
    for some $\phi \in [0, 2\pi]$.   The algorithm runs in $\poly(n, 1/\epsilon, \log(1/\delta), 1/\theta)$.
        \label{th:pauli_estimate}
\end{theorem}

\begin{proof}
    Running Algorithm \ref{alg:find_coef} with parameters $\theta, \epsilon, \delta/2$ uses 
    $$O\left(\frac{\log(1/\theta) + \log(1/\delta)}{\theta^2}\right)$$
    queries to $U$ and, with probability at least $1-\delta/2$, returns a set $\CX$ containing all indices $\bfs$ with $|\alpha_\bfs| \geq \theta$ (by Lemma \ref{lm:coupon}). 
    Next, running Algorithm \ref{alg:estimate-unitary} with inputs $U, \CX, \epsilon, \delta/2$ uses
    $$O\left(\frac{\log(|\CX|/\delta)}{\epsilon^2}\right) = O\left(\frac{\log(1/\theta) + \log(1/\delta)}{\epsilon^2}\right)$$
    queries to $U$ and, with probability at least $1-\delta/2$, estimates each Pauli coefficient $\alpha_\bfs$ for $\bfs \in \CX$ with additive error at most $\frac{3\sqrt{\epsilon}}{|{\alpha}_{\bft}|-\sqrt{\varepsilon}} $ (by Lemma \ref{lm:pauli_error}).
    By the union bound, both algorithms succeed with probability at least $1-\delta$.
    Finally, since $|\alpha_{\bft}| \geq |\hat{\alpha}_{\bft}| - \sqrt{\epsilon}$ (by the triangle inequality) and $|\hat{\alpha}_{\bft}| \geq \theta - \sqrt{\epsilon}$ (by Lemma \ref{lm:coupon}), choosing $\sqrt{\epsilon} \leq \theta/4$ ensures $|\alpha_{\bft}| \geq \theta - 2\sqrt{\epsilon} \geq \theta/2$. Thus, the additive error in estimating each $\alpha_\bfs$ is at most $\frac{3\sqrt{\epsilon}}{|{\alpha}_{\bft}|-\sqrt{\varepsilon}}  \leq \frac{12\sqrt{\epsilon}}{\theta}$. Changing the variable $\epsilon$ to $(\tfrac{\epsilon}{12})^2$ completes the proof. The running time is polynomial in all parameters since both algorithms run in polynomial time. Particularly,  as  $M_\bfs, R_{\bft,\bfs}, I_{\bft,\bfs}$ belong to the Clifford group, the running time of \CST  is polynomial.

\end{proof}

Although the algorithm estimates the Pauli coefficients up to a global phase, this ambiguity does not affect the final learning guarantee. As established by the Generalized Distance Bound (Lemma \ref{lm:gen_l1_to_diamond}), the diamond distance is strictly upper-bounded by the phase-aligned operator distance.
\section{Learning Unitary Operations}\label{sec:Ulearning}
The objective of the learner is to approximate an unknown unitary operator $U$
acting on $n$ qubits such that the learned unitary $\hat{U}$ behaves similarly to
$U$ on relevant input states.

Accordingly, the central question is:
\emph{what is an appropriate metric to quantify the closeness between a unitary $U$
and its approximation $\hat{U}$ in order to ensure that $\hat{U}$ yields outputs
similar to those of $U$?}

The diamond distance (see Definition \ref{def:diamond}) is a standard measure of distinguishability between quantum
channels \citep{Wilde2013}. If the diamond distance between two channels is small,
then their outputs are close in trace distance for all possible input states.

Two key challenges need to be addressed: (1) how to connect the diamond distance to
the Pauli coefficients of the unitary, estimated per Theorem \ref{th:pauli_estimate}; and   (2) how to efficiently construct an estimate $\hat{U}$ to $U$ from the estimated Pauli coefficients in polynomial time.
We first address the second challenge, deferring the first to the next subsection.

\subsection{Efficient Implementation via Approximate LCU}
Suppose we have identified  $\hat{U}= \sum_{\bfs \in \CX} \hat{\alpha}_{\bfs} \sigma_{\bfs}$ as an sparse approximation to the target unitary $U$.

$\hat{U}$ need not be exactly unitary. But we can  ``promote'' $\hat{U}$ to an implementable (near-)unitary using the \ac{LCU} technique. 
Suppose that $\hat{U}$ is $\gamma$-close to $U$ in Pauli $\ell_1$ norm, i.e., 
 
\[
a := \sum_{\bfs\in S} |\hat\alpha_{\bfs}|\qquad\text{and}\qquad
\gamma := \norm*{U-\hat{U}}_{1, P}.
\]

\begin{observation}[Approximate LCU, c.f. Thm 2.5 of \citep{Kothari2014}]\label{lm:LCU}
There exists a quantum algorithm that implements $\hat{U}$ as a circuit $W$ with error $O(a\sqrt{\gamma})$ in operator norm.  Furthermore, the algorithm runs in $\poly(a, n, |\CX|)$ time.
\end{observation}
\begin{proof}
    The proof follows from   
    
    Theorem \ref{thm:approx LCU}.

    The Approximate LCU construction requires our estimator $\hat{U}$ to be close to the target unitary $U$ in operator norm. Since $\hat{U}$ is $\gamma$-close to $U$ in Pauli $\ell_1$-norm, it is strictly bounded by the same distance in operator norm ($\|\hat{U} - U\|_{op} \le \|\hat{U} - U\|_{1,P} \le \gamma$). Therefore, we can apply the Approximate LCU construction, which yields a unitary circuit $V$ that approximates $U$ up to an error of $O(a\sqrt{\gamma})$ in operator norm. Finally, because both the target $U$ and the implemented circuit $V$ are strictly unitary, we apply the exact unitary bound from Lemma \ref{lm:gen_l1_to_diamond}. The diamond distance is strictly bounded by $2\|U - V\|_{op}$, yielding an overall diamond distance error of $O(a\sqrt{\gamma})$.
    The algorithm relies on the state-preparation unitary 
    \[
       A \ket{0^m} =  \sum_{\bfs \in \CX} \sqrt{\frac{|\hat{\alpha}_\bfs|}{a}}\ket{\bfs},
    \]
    where $m=\lceil \log_2 |\CX| \rceil$,  which can be implemented using $\Theta(2^m)=\Theta(|\CX|)$ gates \citep{Moettoenen2005}. 
    In addition, we need to implement the select unitary 
\[
         V = \sum_{\bfs \in \CX} \ketbra{\bfs} \otimes \sigma_{\bfs},
\]
This can be done efficiently using $O(n)$ gates per Pauli operator, and hence the overall complexity is polynomial in $n$ and $|\CX|$.
\end{proof}

With this lemma, we just need to control the parameters $a$ and $\gamma$ to ensure an efficient implementation of a unitary close to $U$. In the next section, we show that for nearly sparse unitaries, $\CX$, $a$ and $\gamma$ can be bounded polynomially in the sparsity $s$. This follows from our previous algorithms for support identification and coefficient estimation.

Building upon this, we summarize the learning procedure as Algorithm \ref{alg:super}. 

\begin{algorithm2e}[ht]
\DontPrintSemicolon
\LinesNumbered
\caption{Unitary Learning and Implementation}
\label{alg:super}

\KwIn{Query access to unitary $U$; accuracy $\epsilon \in (0,1)$, failure probability $\delta \in (0,1)$.}
\KwOut{A block-encoded unitary $W$ implementing $\hat{U}$.}

\BlankLine

Run Algorithm \ref{alg:find_coef} with $U$ , $\theta = \epsilon/\sqrt{s}$, $\epsilon$, and $\delta$ to obtain support set $\mathcal{X}$ and the largest coefficient $\hat{\alpha}_\bft$\;

Run Algorithm \ref{alg:estimate-unitary} with unitary $U$, $\mathcal{X}$, $\delta$, and accuracy $\varepsilon = O(\epsilon^2/s^{3})$ to obtain complex coefficients $\{\hat{\alpha}_\bfs\}_{\bfs \in \mathcal{X}}$\;
\BlankLine

Let $m=\lceil \log_2 |\CX| \rceil$ and $\alpha \gets \sum_{\bfs \in \CX} |\hat{\alpha}_\bfs|$. \;
Construct  the prepare oracle $A: \ket{0^m} \mapsto \sum_{\bfs \in \mathcal{X}} \sqrt{\frac{|\hat{\alpha}_\bfs|}{\alpha}} \ket{\bfs}$\;
Construct the select unitary $V = \sum_{\bfs \in \mathcal{X}} \ketbra{\bfs} \otimes \frac{\hat{\alpha}_\bfs}{|\hat{\alpha}_\bfs|} \sigma^\bfs$\;
Run the Approximate LCU procedure (Lemma \ref{lm:LCU}) with inputs $A$, $V$ to obtain the block-encoded unitary $W$ approximating $U$.\;

\BlankLine

\Return the mapping 
$$\ket{\psi} \mapsto W \ket{0^m}\ket{\psi}$$
as the approximation of $U$.\;

\end{algorithm2e}

\subsection{Learning Nearly Sparse Unitaries} 

The special class of unitaries we study here are nearly sparse unitaries as defined in Definition \ref{def:nearly sparse}.

\begin{theorem}
    There exists an algorithm that learns a mapping close to 
    the nearly $(s,\epsilon)$-sparse $n$-qubit unitary $U$ in diamond distance, with $O(\epsilon)$ error, using $\tilde{O}(\frac{s^6}{\epsilon^4 } )$ queries.
    Moreover, with probability greater than $1-\delta$, a circuit can be implemented using the same number of queries, $\poly(n,s,1/\epsilon)$ time,  and $\poly(n,s,1/\epsilon)$ gates,  with additive error $O(\sqrt{s\epsilon})$ in diamond distance to the target unitary.
    \label{th:learning_n_sparse}
\end{theorem}

\begin{proof}
        From Lemma \ref{lm:gen_l1_to_diamond}, without loss of generality, we can assume the global phase ambiguity is fixed ($\phi =0$ in Lemma \ref{lm:pauli_error}).
        We first prove the bound on the query complexity. From Definition \ref{def:nearly sparse},  there is a set $S$ of size s such that $\sum_{\bfs \notin S} |\alpha_\bfs| \leq \epsilon$.  
    Let $\tilde{U} = \sum_{\bfs \in S} \alpha_\bfs \sigma_\bfs$ be the sparse truncation of $U$ to the support $S$.

    Using this and the connection between the Pauli $\ell_1$-norm and $\ell_2$-norm, 
    we have $\norm*{U-\tilde{U}}_{2,P} \leq \norm*{U-\tilde{U}}_{1,P} \le \epsilon$,
    that implies 
    \begin{align*}
        \norm*{\tilde{U}}_{2,P} \ge \norm*{U}_{2,P} - \norm*{U-\tilde{U}}_{2,P} \geq 1 - \epsilon.
        \end{align*}
    Since, $\tilde{U}$ is $s$-sparse, the maximum coefficient $\alpha_{\bft}$ must satisfy $$|\alpha_{\bft}| \ge \frac{1-\epsilon}{\sqrt{s}},$$ ensuring it is distinguishable from zero.
    
    We apply Algorithm \ref{alg:find_coef} with $\theta=\frac{\varepsilon}{\sqrt{s}}$ to get the target set $\CX$, as $m_1 = \tilde{O}(\frac{s}{\varepsilon^2})$ we have $\abs{\CX} = \tilde{O}(\frac{s}{\varepsilon^2})$.
    Then, we estimate the Pauli coefficients indexed by $\CX$ using Algorithm \ref{alg:estimate-unitary} with accuracy parameter $\varepsilon$. From Lemma \ref{lm:pauli_error}, Algorithm \ref{alg:estimate-unitary} estimates $\alpha_\bfs$ with precision 
    $$\epsilon_2 = \frac{3\sqrt{\varepsilon}}{\abs{\alpha_\bft}-\sqrt{\varepsilon}} \leq  \frac{3\sqrt{\varepsilon}}{\frac{1-\epsilon}{\sqrt{s}}-\sqrt{\varepsilon}} = \frac{3\sqrt{s\varepsilon}}{1-\epsilon-\sqrt{s\varepsilon}}
    $$
    Then, we zero out any coefficient with magnitude below $2\epsilon_2$; generating the set $$\CX_{\epsilon_2} = \{\bfs\in \CX: |\hat{\alpha}_\bfs|\geq {2\epsilon_2}\}.$$ With probability at least $1-\delta$, $\CX_{\epsilon_2} \subseteq S$; because the large coefficients of $U$ are all contained in $S$, and the estimation error is smaller than $\epsilon_2$. Moreover, for each $\bfs \in S \setminus \CX_{\epsilon_2}$,
    $$
    \abs{\alpha_\bfs} \le \abs{\hat{\alpha}_\bfs} + \epsilon_2 \le 3\epsilon_2.
    $$
    Let the resulting estimated unitary be $$\hat{U} = \sum_{\bfs \in \CX_{\epsilon_2}} \hat{\alpha}_\bfs \sigma_\bfs.$$

    The total estimation error is bounded by the triangle inequality:
    \begin{align*}
        \norm{U-\hat{U}}_{1,P} &\leq \norm{U-\tilde{U}}_{1,P} + \norm{\tilde{U}-\hat{U}}_{1,P}\\
        & \leq \epsilon + \sum_{\bfs \in S\cap \CX_{\epsilon_2}} \abs{\alpha_\bfs - \hat{\alpha}_\bfs} + \sum_{\bfs \in S/ \CX_{\epsilon_2}} |\alpha_\bfs|\\
        & \leq \epsilon + |S\cap \CX_{\epsilon_2}| \epsilon_2 + 3|S/ \CX_{\epsilon_2}| \epsilon_2\\\numberthis \label{eq:final_error}
        & \leq \epsilon + 3 s \epsilon_2 = \epsilon + O(\sqrt{\varepsilon} s^{3/2}).
    \end{align*}
    Hence, to ensure the $\ell_1$ norm error is $2\epsilon$, we require the accuracy parameter $\varepsilon = O(\epsilon^2/s^{3})$. Thus, the required number of samples $m_2$ in Algorithm \ref{alg:estimate-unitary} is:
    $$ 
    m_2 = \tilde{O}\left(\frac{s^6}{\epsilon^4 } \log\left(\frac{s}{\delta \epsilon}\right)\right)
    $$
   
    By Corollary \ref{cor:gen_l1_to_dim}, the diamond distance between the true unitary channel $\CU$ and the map $\hat{\CU}$ induced by our non-unitary estimator is bounded by $2(2\epsilon) + (2\epsilon)^2 = O(\epsilon)$. This concludes the query complexity part of the theorem.

    Lastly, we construct a unitary channel by applying LCU to implement $\hat{U}$ as a strictly unitary mapping $W$. This incurs an additional error $O(a\sqrt{\gamma})$ in operator norm per Observation \ref{lm:LCU}. In our case, $\gamma = 2\epsilon$ and
    $$a = \sum_{\bfs \in \CX_{\epsilon_2}} |\hat{\alpha}_\bfs| \leq \sqrt{s} \norm{\hat{U}}_{2,p}\leq   \sqrt{s}(1+\epsilon),$$
 where we used the relation between $\ell_1$ and $\ell_2$ bound and the fact that $\hat{U}$ is  $s$-sparse. 
    Thus, the overall error  is bounded by:
    \begin{align*}
        \|W - U\|_{op} &\le \|W - \hat{U}\|_{op} + \|\hat{U} - U\|_{op} \\
        &\le O(\sqrt{s\epsilon}) + \|\hat{U} - U\|_{1,P} \\
        &= O(\sqrt{s\epsilon}) + 2\epsilon = O(\sqrt{s\epsilon}).
    \end{align*}
    Because both $W$ and $U$ are exactly unitary, we can now safely apply Lemma \ref{lm:gen_l1_to_diamond}. The diamond distance between the target channel $\mathcal{U}$ and the implemented circuit $\mathcal{W}$ is strictly bounded by:
    $$ d_\diamond(\mathcal{W}, \mathcal{U}) \le 2 \|W - U\|_{op} = O(\sqrt{s\epsilon}). $$
    The overall time complexity is polynomial in $n$, $s$ and $1/\epsilon$.
   
\end{proof}

\begin{corollary}[Improved complexity for unitaries with dominant Pauli coefficient]
    If in addition to being nearly $(s, \epsilon)$ sparse in the conditions of Theorem \ref{th:learning_n_sparse}, the target unitary $U$ has a dominant Pauli coefficient satisfying $|\alpha_{\bft}| \geq c$ for some constant $c>1/100$. Then the query complexity reduces to $\tilde{O}(\frac{s^4}{\epsilon^4 }) $. 
    \label{cor:learning_strongly_n_sparse}
\end{corollary}
\begin{proof}
    The proof follows the same steps as Theorem \ref{th:learning_n_sparse}, except that the precision for estimating each coefficient is now $\epsilon_2 = O(\sqrt{\varepsilon})$ since $|\alpha_{\bft}|$ is lower bounded by a constant. Thus, the total error is bounded by
    \begin{align*}
        \norm{U-\hat{U}}_{1,P}  \leq \epsilon + 2 s \epsilon_2 = O(\varepsilon s).
    \end{align*}
    
     Hence, to ensure the $\ell_1$ norm error is $2\epsilon$, we require the accuracy parameter $\varepsilon = O((\epsilon/s)^2)$.
    Thus, the required number of samples $m_2$ in Algorithm \ref{alg:estimate-unitary} is:
    $$
    m_2 = \tilde{O}\left(\frac{s^4}{\epsilon^4 }\right).
    $$

\end{proof}

\begin{corollary}[Improved complexity for bounded-norm unitaries]
    If in addition to being nearly $(s, \epsilon)$-sparse in  Theorem \ref{th:learning_n_sparse},  $\avec{U}{1}$ is bounded by a constant, then  exists a polynomial time algorithm that  learns $U$ with $O(\sqrt{\epsilon})$ error in diamond distance using
    $\tilde{O}(\frac{s^6}{\epsilon^4 }) $. If additionally, $|\alpha_{\bft}| \geq c$ for some constant $c>1/100$, then the  algorithm needs  $\tilde{O}(\frac{s^4}{\epsilon^4 }) $ queries.
    \label{cor:learning_bounded_n_sparse}
\end{corollary}
\begin{proof}
    The proof follows the same steps as Theorem \ref{th:learning_n_sparse}, except that the parameter $a$ in the Approximate LCU procedure is now bounded by $a \leq L_1$.

\end{proof}

\section{Learning Unitaries with bounded \texorpdfstring{$\ell_1$}{Lg}-norm} \label{sec:l1_bounded}
Even when the target unitary or Hamiltonian is not itself sparse, sparse unitaries can still play a crucial role as good approximations or estimators.
We still seek to study learnability of unitaries with bounded Pauli $\ell_1$-norm. Our lower bound results in Section \ref{sec:lb} show that bounding the Pauli $\ell_1$-norm alone does not suffice to ensure efficient learnability.  In this section, we show that the learnability holds under a more relaxed notion of closeness in diamond norm, which we will define later.

    We first motivate this notion from classical learning theory in Section~\ref{sec:classical}, and then formally define it in Section~\ref{sec:relaxed_distance}.

\subsection{Motivation from Classical Learning Theory}\label{sec:classical}

The motivation for our learning problem is rooted in the classical framework of learning Boolean functions via \emph{membership queries} (MQ), which has been extensively studied in learning theory \citep{KM1993,Valiant1984,Kearns1994,Feldman2012}. In the MQ model, a learner is given oracle access to an unknown target function $f:\{0,1\}^n \to \{0,1\}$ and may adaptively query the value of $f(\bfx)$ on inputs $\bfx$ of its choosing. The learner’s goal is to output a hypothesis $h$ whose predictions agree with those of $f$ on most inputs drawn from a fixed distribution $D$, namely,
$$
\Pr_{\bfX \sim D}\!\left[h(\bfX) \neq f(\bfX)\right] \le \varepsilon .
$$

A key feature of this model is that learning is evaluated \emph{distributionally}: the hypothesis $h$ is not required to approximate $f$ uniformly over all inputs, but only to perform well on typical inputs sampled from $D$. This viewpoint is both theoretically and practically significant, as the distribution $D$ is often chosen to reflect the inputs encountered in realistic settings.

The MQ model has led to a rich body of work characterizing the learnability of various concept classes under specific distributions, such as the uniform or product distributions \citep{KM1993,Feldman2012,HeidariK2025}. These results provide a foundational perspective for our work, in which we seek an analogous distributional notion of learnability for quantum operations accessed through query models.

\subsection{Relaxed Notions of Closeness for Learning Bounded \texorpdfstring{$\ell_1$}{Lg}-Norm Unitaries}\label{sec:relaxed_distance}

Additionally, our motivation is that, in many practical settings, one is not concerned with the action of a quantum channel on \emph{all} possible input states, but rather with its behavior on a specific subset of states of interest. This naturally leads to a relaxed distinguishability measure that evaluates closeness only on such restricted inputs.

\begin{definition}[Restricted diamond distance]
Let $\CU$ and $\hat{\CU}$ be the quantum channels induced by unitary $U$ and linear mapping $\hat{U}$ acting on a system $A$, and let $\CA$ be a specified set of states on $A$. The \emph{restricted diamond distance} between $\CU$ and $\hat{\CU}$ with respect to $\CA$ is defined as
\begin{equation}
d_{\diamond,\CA}(\CU,\hat{\CU})
:=
\sup_R \;
\sup_{\rho_{RA}:\, \Tr_R(\rho_{RA}) \in \CA}
\norm{
(\II_R \otimes \CU)(\rho_{RA})
-
(\II_R \otimes \hat{\CU})(\rho_{RA})
}_1 ,
\end{equation}
where the supremum is taken over all finite-dimensional reference systems $R$.
\end{definition}

Under this definition, the distinguishability of two unitaries is evaluated only on those joint input states whose marginal on the system of interest lies in $\CA$. Thus, $d_{\diamond,\CA}$ quantifies the worst-case deviation between $\CU$ and $\hat{\CU}$ when restricted to the relevant class of inputs. An example is the case where $\CA$ consists of all product states on $A$. In this case, the restricted diamond distance measures how well the two channels can be distinguished when the input states are unentangled with any reference system.

At one extreme, when $\CA$ consists of all quantum states on $A$, the restricted diamond distance reduces to the standard diamond distance. At the other extreme, choosing $\CA$ to be a highly structured or low-dimensional family of states yields a significantly weaker, but often more operationally meaningful notion of closeness.

An alternative and operationally motivated way to define closeness between unitaries is through the output probability distributions obtained after measurement. This perspective is natural in learning settings, as it captures the requirement that replacing the true unitary $U$ by a hypothesis $\hat{U}$ inside a larger quantum procedure should not significantly alter the algorithm’s observable outcomes.

\begin{definition}[$\varepsilon$-indistinguishability on a restricted input set]
Let $\CU$ and $\hat{\CU}$ be the unitary channels induced by unitaries $U$ and $\hat{U}$ acting on a system $A$, and let $\CA$ be a specified set of input states on $A$. We say that $\CU$ and $\hat{\CU}$ are \emph{$\varepsilon$-indistinguishable on $\CA$} if
\begin{equation}
\sup_R \;
\sup_{\rho_{RA}:\, \Tr_R(\rho_{RA}) \in \CA}
\;
\sup_{\CM_{RA}}
d_{\mathrm{TV}}\!\left(P_{\CM,\rho}, \hat{P}_{\CM,\rho}\right)
\le \varepsilon ,
\label{eq:wilde_dist}
\end{equation}
where $R$ is an arbitrary finite-dimensional reference system, $\CM_{RA}$ ranges over all measurements on the joint system $RA$, and $P_{\CM,\rho}$ and $\hat{P}_{\CM,\rho}$ denote the outcome distributions obtained by measuring
\[
(\II_R \otimes \CU)(\rho_{RA})
\quad\text{and}\quad
(\II_R \otimes \hat{\CU})(\rho_{RA}),
\]
respectively, using the measurement $\CM_{RA}$.
\end{definition}

It is worth noting that the definition of $\varepsilon$-indistinguishable unitaries

aligns with the requirement of faithful simulation of quantum measurements studied in quantum information theory \citep{wilde2012information,Atif2022,winter2004extrinsic}.

The relationship between $\varepsilon$-indistinguishability and the restricted diamond distance is formalized by the following lemma.

\begin{lemma}
If
$
d_{\diamond,\CA}(\CU,\hat{\CU}) \le \varepsilon ,
$
then $\CU$ and $\hat{\CU}$ are $\varepsilon$-indistinguishable on $\CA$, i.e., \eqref{eq:wilde_dist} holds.
\label{lm:def_eq}
\end{lemma}
\begin{proof}
    Using Lemmas 2 and 4 from \citep{wilde2012information}, for any purification $\phi_\rho$ of $\rho_A \in \mathcal{A}$, the condition in \eqref{eq:wilde_dist} is equivalent to
$$\norm{(\II \otimes \CU)(\phi_\rho) - (\II \otimes \hat{\CU})(\phi_\rho)}_1 \leq \varepsilon.$$

    Since the trace norm is contractive under completely positive trace-preserving
(CPTP) maps, any subsequent measurement $\mathcal{M}$ on $RA$ can only decrease
the distance:
    \[
        d_{\mathrm{TV}}\!\big(P_{\mathcal{M},\rho}, \hat{P}_{\mathcal{M},\rho}\big)
        \le \tfrac{1}{2}
        \big\| (\II_R \otimes \mathcal{U})(\phi_\rho)
          - (\II_R \otimes \hat{\mathcal{U}})(\phi_\rho) \big\|_1.
    \]
    Taking the supremum over all $R$, $\rho_{RA}$ with $\Tr_R(\rho_{RA}) \in \mathcal{A}$,
    and over all POVMs $\mathcal{M}$, we obtain
    \[
        \sup_R \sup_{\mathcal{M}_{RA}}
        \sup_{\rho_{RA}: \Tr_R(\rho_{RA}) \in \mathcal{A}}
        d_{\mathrm{TV}}\!\big(P_{\mathcal{M},\rho}, \hat{P}_{\mathcal{M},\rho}\big)
        \le d_{\diamond,\mathcal{A}}(\mathcal{U}, \hat{\mathcal{U}}).
    \]
    Therefore, if $d_{\diamond,\mathcal{A}}(\mathcal{U}, \hat{\mathcal{U}}) \le \varepsilon$, then
    \eqref{eq:wilde_dist} follows.
\end{proof}

\subsection{Connection to \texorpdfstring{$\ell_2$}{Lg}-norm}
For unitary channels, the diamond distance admits dimension-dependent upper and lower bounds in terms of the Hilbert--Schmidt inner product (equivalently, the entanglement infidelity). In particular, from \citep[Proposition~1.9]{Haah2023}, for unitaries $U$ and $\hat{U}$ acting on an $N$-dimensional Hilbert space,
$$
2\sqrt{1 - \frac{|\Tr(U^\dagger \hat{U})|^2}{N^2}}
\le
d_\diamond(\mathcal{U}, \hat{\mathcal{U}})
\le
2\sqrt{N}\sqrt{1 - \frac{|\Tr(U^\dagger \hat{U})|^2}{N^2}} .
$$
These inequalities show that, while the diamond distance and Hilbert--Schmidt overlap are closely related for unitary channels, the relationship necessarily incurs a dimension-dependent loss.

We now establish a sharper characterization for the restricted diamond distance, which avoids this dimensional dependence and directly relates the distance to expectation values of the overlap operator with respect to admissible input states.

Here, we consider the specific instance where $\CA = \set{\II/N}$, the maximally mixed state. This case is of particular importance as it corresponds to the performance of the channel averaged uniformly over all pure input states.

 \begin{lemma}
    Let $\CU$ be the unitary channel for unitary $U$, and $\hat{\CU}$ be the channel for operator $\hat{U}$. If $\norm{U-\hat{U}}_{2,P} = \nu$, then 
    $$
    d_{\diamond, \{\II/N\}}(\CU, \hat{\CU}) \leq 2\nu + \nu^2 .
    $$
    \label{lm:l2tol1}
\end{lemma}
\begin{proof}
    By the convexity of the trace distance, it suffices to restrict the optimization in the definition of $d_{\diamond, \{\II/N\}}$ to pure states acting on a joint system $RA$. Moreover, since the marginal on system $A$ is fixed to $\II/N$, any two purifications $\ket{\phi}_{R_1A}$ and $\ket{\psi}_{R_2A}$ (assuming without loss of generality that $\dim(\mathcal{H}_{R_1}) \le \dim(\mathcal{H}_{R_2})$) are related by an isometry on the reference system \citep{Wilde2013}:
    $$
     \ket{\psi}_{R_2A} = (W \otimes \II_A) \ket{\phi}_{R_1 A} ,
    $$
    where $W: \mathcal{H}_{R_1} \to \mathcal{H}_{R_2}$ is an isometry. Let $E = \II_{R_1} \otimes (\CU - \hat{\CU})$ be the difference of the channels. Then,
    $$
    \norm{E \ketbra{\psi}_{R_2A} }_1 = \norm{E (W\ketbra{\phi}_{R_1A} W^\dagger) }_1 = \norm{W E(\ketbra{\phi}_{R_1A}) W^\dagger}_1 = \norm{E(\ketbra{\phi}_{R_1A})}_1
    $$
    where the last equality follows from the isometric invariance of the trace norm. Therefore, it suffices to evaluate the distance using the canonical maximally entangled state $\ket{\phi} = \frac{1}{\sqrt{N}}\sum_{i=0}^{N-1} \ket{i}_{R_1} \ket{i}_{A}$. We can write the restricted diamond distance as:
    $$
    d_{\diamond, \{\II/N\}}(\CU, \hat{\CU}) = \norm{E(\ketbra{\phi})}_1 .
    $$
    
    Let $\ket{\phi_\CU} = (\II \otimes U) \ket{\phi}$ and define $\ket{\phi_{\hat{\CU}}}$ analogously. Letting $\ket{\delta} = \ket{\phi_\CU} - \ket{\phi_{\hat{\CU}}}$, we evaluate its squared norm:
    \begin{align*}
        \norm{\ket{\delta}}_2^2 &= \bra{\phi} \left( \II \otimes (U-\hat{U})^\dagger(U-\hat{U}) \right) \ket{\phi} \\
        &= \frac{1}{N} \sum_{i,j=0}^{N-1} (\bra{i}_R \bra{i}_A) \left( \II_R \otimes (U-\hat{U})^\dagger(U-\hat{U}) \right) (\ket{j}_R \ket{j}_A) \\
        &= \frac{1}{N} \sum_{i=0}^{N-1} \bra{i}_A (U-\hat{U})^\dagger(U-\hat{U}) \ket{i}_A \\
        &= \frac{1}{N} \Tr\left((U-\hat{U})^\dagger(U-\hat{U})\right) \\
        &= \norm{U-\hat{U}}_{2,P}^2 = \nu^2 .
    \end{align*}
    Thus, $\norm{\ket{\delta}}_2 = \nu$. Finally, expanding the density operators, we bound the trace distance:
    \begin{align*}
        \norm{\ketbra{\phi_\CU} - \ketbra{\phi_{\hat{\CU}}}}_1 &= \norm{\ketbra{\phi_\CU} - (\ket{\phi_{\CU}} - \ket{\delta})(\bra{\phi_{\CU}} - \bra{\delta})}_1 \\
        &= \norm{\ket{\phi_\CU} \bra{\delta} + \ket{\delta}\bra{\phi_\CU} - \ketbra{\delta}}_1 \\
        &\le 2 \norm{\ket{\phi_\CU} \bra{\delta}}_1 + \norm{\ketbra{\delta}}_1 \\
        & = 2\nu + \nu^2 .
    \end{align*}
    The final equality relies on the fact that $\norm{\ket{\psi_1}\bra{\psi_2}}_1 = \norm{\ket{\psi_1}}_2 \norm{\ket{\psi_2}}_2$, and that $\ket{\phi_\CU}$ is perfectly normalized since $U$ is unitary.
\end{proof}

\subsection{Learning implications}
First, we introduce a useful lemma that connects the Pauli $\ell_1$-norm of a unitary to its sparse approximability in Pauli $\ell_2$-norm. This is motivated by similar results in learning Boolean functions with bounded spectral norm \citep{KM1993,HeidariK2025}. We show that if a unitary has a bounded Pauli $\ell_1$-norm, then it can be well-approximated by a sparse unitary in Pauli $\ell_2$-norm.

    \begin{lemma}
        Suppose an $n$-qubit unitary $U$ is approximated by $\tilde{U}$ such that

        $\norm*{U - \tilde{U}}_{2}^2 \le \varepsilon/4$ and 
        $\norm*{\tilde{U}}_{1,P} = L_1$. 
        Then, $U$ can be approximated in $\ell_2$-norm by a $\tfrac{4L_1^2}{\varepsilon}$-sparse operator $\hat{U}$ generated from the Pauli coefficients $U$ larger than $\tfrac{\varepsilon}{2L_1}$. 
         More precisely, let $$S = \{\bfs\in \{0,1,2,3\}^n \mid |\alpha_{\bfs}| \ge {\frac{\varepsilon}{2L_1}} \}.$$ 
         Then, there exists a set $\mathcal{S}^* \subseteq S$,  with size $|\mathcal{S}^*| \le \tfrac{4L_1^2}{\varepsilon}$, such that the operator 
        \[
            \hat{U} = \sum_{\bfs \in \mathcal{S}^*} \hat{\alpha}_{\bfs} \sigma^{\bfs},
        \]
        with  $|\hat{\alpha}_{\bfs} - \alpha_{\bfs}| \le \frac{\varepsilon}{2L_1}$ for all $\bfs \in \CS^*$, 
        satisfies 
        \[
            \avec{U - \hat{U}}{2}^2 \le 6\varepsilon.
        \]
        \label{lm:app_unitary}
    \end{lemma}
        
    \begin{proof}
        We first show existence of a sparse $V$ that approximates $\tilde{U}$ and hence $U$.
        Write the Pauli expansion of $\tilde{U} = \sum_{\bfs} \tilde{\alpha}_{\bfs} \sigma^{\bfs}$. 
        Let 
        $S' = \{\bfs \mid |\tilde{\alpha}_{\bfs}| \ge \frac{\varepsilon}{4L_1}\}$ 
        and define 
        $V = \sum_{\bfs \in S'} \tilde{\alpha}_{\bfs} \sigma^{\bfs}$. 
        Then $$\avec{V}{1}\le \norm*{\tilde{U}}_{1,P} = L_1$$ and 
        $|S'| \le \frac{4L_1^2}{\varepsilon}$.
        Moreover,
        \[
        \avec{\tilde{U} - V}{2}^2
        = \sum_{\bfs \notin S'} |\tilde{\alpha}_{\bfs}|^2
        \le \max_{\bfs \notin S'} |\tilde{\alpha}_{\bfs}| 
            \sum_{\bfs \notin S'} |\tilde{\alpha}_{\bfs}| 
        \le \frac{\varepsilon}{4L_1} L_1 = \frac{\varepsilon}{4}.
        \]
        Using $\|A+B\|^2 \le 2(\|A\|^2 + \|B\|^2)$, we get
        \[
        \avec{U - V}{2}^2
        \le 2(\avec{U - \tilde{U}}{2}^2 +\avec{\tilde{U} - V}{2}^2)
        \le \varepsilon.
        \]

        It remains to show that the sparse approximation $\hat{U}$ constructed from $U$ satisfies the desired error bound.

        Particularly, we show that  
        $\norm*{U - \hat{U}}_{2,P}^2 \le 3\varepsilon$.
        
        From the above construction, we have
        \[
        \varepsilon \ge \avec{U - V}{2}^2 
        = \sum_{\bfs \notin S'} |\alpha_{\bfs}|^2 
          + \sum_{\bfs \in S'} |\alpha_{\bfs} - \tilde{\alpha}_{\bfs}|^2.
        \]
        implying $\sum_{\bfs \notin S'} |\alpha_{\bfs}|^2 \le \varepsilon$. 
        Let $\CS^* = S \cap S'$ and 
        $K = \sum_{\bfs \in \CS^*} \alpha_{\bfs} \sigma^{\bfs}$. 
        Then
        
        \begin{align*}
            \avec{U - K}{2}^2 &=\sum_{\bfs \notin \CS^*} |\alpha_{\bfs}|^2\\
            &
            = \sum_{\bfs \notin S'} |\alpha_{\bfs}|^2 
              + \sum_{\bfs \in S' - \CS^*} |\alpha_{\bfs}|^2 \\
            &\le \varepsilon + |S'| \cdot \frac{\varepsilon^2}{4L_1^2} \le 2\varepsilon.
        \end{align*}        
        Since $\CS^* \subseteq S',$ $|\CS^*|\leq \tfrac{4L_1^2}{\varepsilon}$. Moreover, as $\gamma \leq \tfrac{\varepsilon}{2L_1}$, the error of approximating $U$ by $\hat{U}$ can be bounded as follows.   
        \begin{align*}
            \avec{U - \hat{U}}{2}^2 \leq 2(\avec{U - K}{2}^2 + \avec{K - \hat{U}}{2}^2)\leq 2(2\varepsilon + |\CS^*| \gamma^2) \leq 6\varepsilon.
        \end{align*}

    \end{proof}

    The lemma indicates that unitaries with bounded Pauli $\ell_1$-norm can be well-approximated by sparse unitaries in Pauli $\ell_2$-norm.

    Immediately, any unitary with bounded Pauli $\ell_1$-norm satisfies the condition of the lemma, and hence can be approximated by a sparse operator in $\ell_2$-norm.

We are ready to present our main result on learning unitaries with bounded Pauli $\ell_1$-norm under the restricted diamond distance.

\begin{theorem}
    The query complexity of learning $n$-qubit unitaries, with Pauli $\ell_1$-norm bounded by $L_1$, and with $O(\epsilon)$ error in the restricted diamond distance $ d_{\diamond, \{\II/N\}}$, is 
    $ \tilde{O}\left( \frac{L_1^8}{\epsilon^{16}} \right). $
    
    \label{th:learning_apx_sparse}
\end{theorem}

\begin{proof}
    We use Lemma \ref{lm:app_unitary} with $\varepsilon = \tfrac{\epsilon^2}{6}$. Since $U$ has bounded Pauli $\ell_1$-norm, the sparse index set is $\CS^*=\CS$ where $\CS = \{\bfs: |\alpha_\bfs|\geq \frac{\varepsilon}{2L_1}\}$. This means if our estimated coefficients satisfy $|\alpha_\bfs - \hat{\alpha}_\bfs| \leq \frac{\varepsilon}{2L_1} =\frac{\epsilon^2}{12L_1}$, then the estimator $\hat{U} := \sum_{\bfs \in \CS} \hat{\alpha}_\bfs \sigma^\bfs$ achieves an error of $\norm*{U - \hat{U}}^2_{2,P} \leq 6 \varepsilon = \epsilon^2$. Hence, $\norm*{U - \hat{U}}_{2,P} \leq \epsilon$.

    To achieve this, we use Theorem \ref{th:pauli_estimate} with threshold $\theta \gets \frac{\varepsilon}{2L_1} = \frac{\epsilon^2}{12L_1}$, and accuracy parameter $\epsilon' \gets \frac{\epsilon^4}{144 L_1^2}$. As a result, there is an algorithm (Algorithm \ref{alg:find_coef} then \ref{alg:estimate-unitary}) that identifies a set $\CX$ containing $\CS$ and estimates $\alpha_\bfs$ with additive error:
    $$ \frac{\epsilon'}{\theta} = \frac{\epsilon^4/ (144 L_1^2)}{\epsilon^2/(12L_1)} = \frac{\epsilon^2}{12L_1} $$
    which satisfies our desired precision. 

    Now, we evaluate the query complexity. Theorem \ref{th:pauli_estimate} states the number of queries scales as $\tilde{O}(1/(\epsilon')^4)$. Substituting our chosen accuracy parameter $\epsilon'$ into the complexity yields:
    $$ T = \tilde{O}\left( \frac{1}{(\epsilon')^2} \right) = \tilde{O}\left( \left(\frac{144 L_1^2}{\epsilon^4}\right)^4 \right) = \tilde{O}\left( \frac{L_1^8}{\epsilon^{16}} \right). $$

    Finally, having established that the estimator satisfies $\norm{U - \hat{U}}_{2,P} \leq \epsilon$, we apply Lemma \ref{lm:l2tol1}. This guarantees that the restricted diamond distance between the true unitary channel $\CU$ and the estimated channel $\hat{\CU}$ acting on the maximally mixed state is bounded by:
    $$ d_{\diamond, \{\II/N\}}(\CU, \hat{\CU}) \leq 2\epsilon + \epsilon^2 = O(\epsilon). $$
    This concludes the proof of the query complexity.
\end{proof}
\section{Lower bounds}\label{sec:lb}
This section establishes lower bounds for learning sparse and nearly sparse unitaries, as well as for unitaries with bounded Pauli $\ell_1$ norm.
\subsection{Lower bound for nearly sparse unitaries}
We first establish a lower bound for learning exactly sparse unitaries, and then extend it to nearly sparse unitaries.
\begin{theorem}
Any algorithm that learns the class of $n$-qubit unitary $U$ with at most $s$ non-zero Pauli coefficients to error $\epsilon$ in diamond distance and with success probability at least $2/3$ must make at least
$\Omega(\frac{s}{\epsilon})$ queries to $U$.
\end{theorem}
 \begin{proof}
Take $k=\lfloor \log_2 s \rfloor$. The problem can be reduced to learning $k$-Pauli-dimensional unitaries. A unitary is said to be $k$-Pauli-dimensional if its Pauli expansion is supported on a Pauli subgroup of size $2^k$. Any such unitary has at most $2^k$ non-zero Pauli coefficients, and is hence $s$-sparse with $s = 2^k$. Therefore, the class of unitaries with at most $s$ non-zero Pauli coefficients contains the class of $k$-Pauli-dimensional unitaries with $2^k \le s$. Thus, the theorem follows from the lower bound of $\Omega(2^k/\epsilon) = \Omega(s/\epsilon)$ for learning $k$-Pauli-dimensional unitaries \citep[Corollary 7.3.]{Grewal2025}.
 \end{proof}
An immediate corollary extends the lower bound to nearly sparse unitaries.
 \begin{corollary}
Any algorithm that learns the class of $n$-qubit nearly $s$-sparse unitary $U$  to error $\epsilon$ in diamond distance and with success probability at least $2/3$ must make at least $\Omega(\frac{s}{\epsilon})$ queries to $U$.
 \end{corollary}

\subsection{Pauli $\ell_1$ bound is not enough}
There is a natural question whether bounding the Pauli $\ell_1$ norm of a unitary suffices to ensure its efficient learnability. Our next result shows that this is not the case and the bounded Pauli $\ell_1$ norm alone does not suffice to ensure efficient learnability.

\begin{theorem}\label{th:lower_l1}
Any algorithm that learns the class of $n$-qubit unitary $U$ with Pauli $\ell_1$ norm at most $L = O(1)$ to error $\epsilon = 1/10$ in diamond distance and with success probability at least $2/3$ must make at least $\Omega(2^n)$ queries to $U$.
\end{theorem}
\begin{proof}
We prove the claim by considering the Oracle family of unitaries. Consider the \emph{phase oracle} (or Grover oracle) family
\begin{equation*}
U_x = I - 2 \ket{x}\!\bra{x},
\qquad
x \in \{0,1\}^n.
\end{equation*}
The Pauli $\ell_1$ norm of each $U_x$ is bounded by a constant independent of $n$. To see this, note that each $U_x$ admits an expansion using only $Z$-type Pauli operators:
\begin{equation*}
\alpha_I = 1 - \frac{2}{N},
\end{equation*}
and for each non-identity Pauli $Z_s$,
\begin{equation*}
\alpha_{Z_s}
=
-\frac{2}{N} (-1)^{s \cdot x},
\end{equation*}
where $s\cdot x$ denotes the mod-$2$ inner product.

Therefore,
\begin{align*}
\norm{U_x}_{P,1}
&=
\abs{1-\frac{2}{N}}
+ (N-1)\frac{2}{N} \\
&=
3 - \frac{4}{N}
< 3 .
\end{align*}

Thus, the family $\{U_x\}_{x\in\{0,1\}^n}$ has uniformly bounded Pauli $\ell_1$ norm with
$
L < 3,
$
independent of the number of qubits $n$.

For $x \neq y$, the unitary channels induced by $U_x$ and $U_y$ are perfectly distinguishable: their diamond distance satisfies
\begin{equation*}
\norm{U_x-U_y}_\diamond = 2 .
\end{equation*}
Consequently, any learner that outputs an estimate $\hat{U}$ such that
\begin{equation*}
\norm{U_x-\hat{U}}_\diamond < 1
\end{equation*}
must, in effect, identify the marked string $x$ exactly. However, determining $x$ given black-box access to $U_x$ is exactly the unstructured search problem and is known to require $\Omega(2^{n/2})$ queries. This lower bound is tight, as Grover's algorithm achieves it and is optimal by the BBBV theorem.
\end{proof}

\section{Special Cases and Examples}\label{sec:examples}
This section discusses some special cases and examples of our results.
\subsection{Hamiltonians with bounded Pauli \texorpdfstring{$\ell_1$}{Lg}-norm}
In this section, we extend our results on the learnability of L1-bounded unitaries in the restricted diamond norm to the context of Hamiltonians and Boolean functions.

\begin{lemma}[Propagation of Pauli $\ell_1$-norm]
    Let $H$ be a Hamiltonian with bounded Pauli $\ell_1$-norm, i.e., $\avec{H}{1} \le L_H$. Then, the unitary evolution $U = e^{iH}$ has a bounded Pauli $\ell_1$-norm given by $\avec{U}{1} \le e^{L_H}$. Consequently, $U$ can be learned using Algorithm \ref{alg:estimate-unitary} with sample complexity dependent on $e^{L_H}$.
    \label{lm:hamiltonianl1}
\end{lemma}

\begin{proof}
    We expand the unitary $U$ using its Taylor series definition:
    $$U = e^{iH} = \sum_{k=0}^{\infty} \frac{i^k}{k!} H^k.$$
    Using the triangle inequality for the Pauli $\ell_1$-norm, we have:
    $$\avec{U}{1} = \avec{\sum_{k=0}^{\infty} \frac{i^k}{k!} {H^k}}
    {1}  \le \sum_{k=0}^{\infty} \frac{1}{k!} \avec{H^k}{1}.$$
    We utilize the submultiplicativity of the Pauli $\ell_1$-norm. For any operators $A$ and $B$, $\avec{AB}{1} \le \avec{A}{1} \avec{B}{1}$. 
    Applying this inductively to $H^k$, we obtain:
    $$\avec{H^k}{1} \le \avec{H}{1}^k \le L_H^k.$$
    Substituting this back into the summation:
    $$\avec{U}{1} \le \sum_{k=0}^{\infty} \frac{L_H^k}{k!} = e^{L_H}.$$
    Since $\avec{U}{1}$ is bounded by $L_{U} = e^{L_H}$, we can invoke Theorem \ref{th:learning_apx_sparse}. To learn $U$ up to restricted diamond distance $\epsilon$, we require a sparse approximation with error parameter $\varepsilon$ satisfying the condition from Lemma \ref{lm:app_unitary}:
    $$\varepsilon^2 \le \frac{\epsilon}{4 L_U^2} = \frac{\epsilon}{4 e^{2L_H}}.$$
    Thus, Algorithm \ref{alg:estimate-unitary} can estimate $U$ by targeting coefficient precision $O(\varepsilon^3)$, ensuring the learned unitary $\hat{U}$ satisfies $\avec{U-\hat{U}}{2} \le \sqrt{\epsilon}$.
\end{proof}

\subsection{Hamiltonians and Boolean functions}
The relationship between classical Boolean functions and diagonal Hamiltonians provides a rich source of examples for bounded-norm unitaries. Given a real-valued function $f:\{0,1\}^n\rightarrow \RR$ define the Hamiltonian \[
    H_f = \sum_x f(x) \ketbra{x},
\]
and its corresponding unitary evolution $U_f=e^{iH_f}$. 

Classes such as DNFs and decision trees have bounded L1 norm \citep{Blum1994,Khardon94}, therefore their corresponding unitary and Hamiltonians have bounded L1 norm. 

These structures appear naturally in combinatorial optimization problems, such as the Quantum Approximate Optimization Algorithm (QAOA). For instance, the Max-Cut Hamiltonian for a graph $G=(V,E)$ is given by:
\[
H_C = \frac{1}{2}\sum_{(i,\,j)\in E(G)} (Z_i Z_j - I),
\]

Here, the Pauli $\ell_1$-norm is proportional to the number of edges. 
Crucially, Lemma \ref{lm:hamiltonianl1} implies that efficient learnability of the unitary $e^{i H_C}$ depends on the quantity $e^{|E|}$. Thus, our algorithm is efficient for these problems where the graph is sparse.

\begin{remark}
    Identify each $Z$-string with a bitmask $z \in \{0,1\}^n$, where $z_j = 1$
indicates a $Z$ acting on qubit $j$. Then, the $Z$-Pauli coefficients of $U_f$  are the Boolean Fourier coefficients of the function $g(x):=e^{if(x)}$:
\[
\alpha_z
= \hat{g}_z:=
2^{-n}
\sum_{x \in \{0,1\}^n}
e^{i f(x)} \, \chi_z(x),
\]
where $z\in \{0,1\}^n$, and  \[
\chi_z(x) := (-1)^{z \cdot x}.
\]
are the characters. 
\end{remark}

\subsection{Examples of nearly Pauli-sparse unitaries}
\label{sec:pauli-compressible-examples}

We present several natural families of strictly nearly  $(s,\varepsilon)$ sparse that are not necessarily $s$-sparse.
Throughout we are interested in the regime $s=\poly(n)$ (typically $s=n^{O(1)}$) while allowing $U$ to have exponentially many nonzero Pauli coefficients.

\subsubsection{A product template: many tiny Pauli rotations}
Consider
\begin{equation}
U \;=\; \prod_{j=1}^{m} e^{-i\theta_j P_j},
\qquad P_j\in\mathcal{P}_n,\;\;\theta_j\in\mathbb{R},
\label{eq:product-template}
\end{equation}
with $m=\poly(n)$.  Expanding each factor as
$e^{-i\theta_j P_j}=\cos\theta_j\, I - i\sin\theta_j\, P_j$
induces a decomposition indexed by subsets $T\subseteq[m]$; terms with $|T|$ ``activated'' rotations carry magnitude proportional to $\prod_{j\in T}|\tan\theta_j|$.

Define $A:=\sum_{j=1}^m |\tan\theta_j|$ and, for $k\in\mathbb{N}$, let $S_k$ denote the set of Pauli strings obtainable as a product of at most $k$ of the $P_j$'s (ignoring phases). Then $|S_k|\le \sum_{r=0}^k\binom{m}{r}$.

\begin{proposition}[Subset-size tail bound]
\label{prop:subset-tail}
For $U$ of the form~\eqref{eq:product-template} and any $k\ge 0$,
\begin{equation}
\sum_{\bfs \notin S_k} |\alpha_\bfs|
\;\le\;
\sum_{r=k+1}^{\infty}\frac{A^r}{r!}
\;\le\;
e^{A}\,\frac{A^{k+1}}{(k+1)!}.
\label{eq:subset-tail}
\end{equation}
Consequently, if $A=O(1)$ and $k=O(1)$ is fixed, then $U$ is nearly $(s,\varepsilon)$-sparse with 
\begin{equation}
s \;=\; \sum_{r=0}^{k}\binom{m}{r} \;=\; O(m^k)=\poly(n),
\qquad
\varepsilon \;\le\; e^{A}\frac{A^{k+1}}{(k+1)!}.
\end{equation}
\end{proposition}

\begin{proof}
Expanding~\eqref{eq:product-template} yields
\begin{equation}
U \;=\; \sum_{T\subseteq[m]} \Big(\prod_{j\in T}(-i\sin\theta_j)\Big)
\Big(\prod_{j\notin T}\cos\theta_j\Big)
\Big(\prod_{j\in T} P_j\Big).
\end{equation}
Each term is associated with a subset $T\subseteq[m]$ of activated rotations. Grouping terms by the resulting Pauli string $\prod_{j\in T} P_j$ gives the Pauli coefficients
\begin{equation}
\alpha_\bfs \;=\; \sum_{\substack{T\subseteq[m]:\\ \prod_{j\in T} P_j = \sigma_\bfs}}
\Big(\prod_{j\in T}(-i\sin\theta_j)\Big)
\Big(\prod_{j\notin T}\cos\theta_j\Big).
\end{equation}
Thus, 
\begin{align*}
\sum_{\bfs \notin S_k} |\alpha_\bfs|
&\le\;
\sum_{r=k+1}^{m} \sum_{\substack{T\subseteq[m]:\\ |T|=r}}
\Big|\prod_{j\in T}(-i\sin\theta_j)\Big|
\Big|\prod_{j\notin T}\cos\theta_j\Big| \\
&\le\;
\sum_{r=k+1}^{m} \sum_{\substack{T\subseteq[m]:\\ |T|=r}}
\prod_{j\in T}|\tan\theta_j| \\
&=\;
\sum_{r=k+1}^{m} \sum_{\substack{T\subseteq[m]:\\ |T|=r}}
\prod_{j\in T}|\tan\theta_j| \\
&=\;
\sum_{r=k+1}^{m} \frac{1}{r!}\Big(\sum_{j=1}^m |\tan\theta_j|\Big)^r \\
&=\;
\sum_{r=k+1}^{m} \frac{A^r}{r!}
\;\le\;
\sum_{r=k+1}^{\infty} \frac{A^r}{r!},
\end{align*}
which establishes the first inequality in~\eqref{eq:subset-tail}. The second inequality follows from the Taylor remainder bound for $e^{A}$.
\end{proof}
Each of the following yields (typically) exponentially many nonzero Pauli coefficients, yet nearly $(s,\varepsilon)$-sparse with $s=\poly(n)$ by Proposition~\ref{prop:subset-tail} (for fixed $k$). 

\begin{remark}[Clifford conjugation]
  When $U$ is nearly $(s,\varepsilon)$-sparse, then so is $U':=CUC^\dagger$ for any $C\in\mathsf{Cl}_n$, with the same parameters $s,\varepsilon$. This follows from the fact that Clifford conjugation permutes Pauli strings (up to sign):
\end{remark}

\begin{example}
[Weak local fields (product unitaries).]
Consider
\begin{equation}
U \;=\; \exp\!\Big(-i\frac{\alpha}{n}\sum_{i=1}^n Z_i\Big)
\;=\; \bigotimes_{i=1}^n e^{-i(\alpha/n)Z_i}.
\end{equation}
Here $m=n$ and $A\asymp |\alpha|$ (for $|\alpha|/n\ll 1$).
\end{example}

\subsubsection{Short-time evolution under Pauli-sparse Hamiltonians}
Let
\begin{equation}
H \;=\; \sum_{j=1}^{m} h_j \sigma^{\bfs_j},
\qquad m=\poly(n),
\label{eq:pauli-sparse-H}
\end{equation}
and define $L:=\sum_{j=1}^m |h_j|$.
For $U:=e^{-itH}$, consider the degree-$k$ truncation
$U_{\le k}:=\sum_{\ell=0}^{k} \frac{(-it)^\ell}{\ell!} H^\ell$.

\begin{proposition}
\label{prop:taylor-tail}
For $U=e^{-itH}$ with $H$ as in~\eqref{eq:pauli-sparse-H},
\begin{equation}
\|U-U_{\le k}\|_{P,1}
\;\le\;
\sum_{\ell=k+1}^{\infty}\frac{(|t|L)^\ell}{\ell!}
\;\le\;
e^{|t|L}\,\frac{(|t|L)^{k+1}}{(k+1)!}.
\label{eq:taylor-tail}
\end{equation}
Moreover, $U_{\le k}$ is supported on at most $\sum_{\ell=0}^{k} m^\ell = O(m^k)$ Pauli strings. In particular, for $|t|L=O(1)$ and fixed $k$, $U$ is $(s,\varepsilon)$-Pauli-compressible with $s=\poly(n)$.
\end{proposition}
\begin{proof}
From  the submultiplicativity of the Pauli $\ell_1$-norm and the expansion of $U-U_{\le k}$, we have
\begin{align*}
\|U-U_{\le k}\|_{P,1}
&=\;\Big\|\sum_{\ell=k+1}^{\infty} \frac{(-it)^\ell}{\ell!} H^\ell \Big\|_{P,1} \\
&\le\;\sum_{\ell=k+1}^{\infty} \frac{|t|^\ell}{\ell!} \|H^\ell\|_{P,1} \\
&\le \sum_{\ell=k+1}^{\infty} \frac{(|t|L)^\ell}{\ell!},
\end{align*}
establishing the first inequality in~\eqref{eq:taylor-tail}. The second inequality follows from the Taylor remainder bound for $e^{|t|L}$.
\end{proof}

\begin{example}
  [Weak commuting Ising layer (bounded-degree graphs).]
For a graph $G=(V,E)$ with $|V|=n$ and $|E|=\Theta(n)$,
\begin{equation}
U \;=\; \exp\!\Big(-i \sum_{(i,j)\in E} J_{i,j}Z_iZ_j\Big),
\end{equation}
so $m=|E|=\Theta(n)$ and $A\asymp |\alpha|$ (for $|\alpha|/n\ll 1$).
  
\end{example}

\subsection{Examples of unitaries with small Pauli \texorpdfstring{$\ell_1$}{Lg}-norm}\label{sec:exp}
\label{sec:examples l1 small}
In this section, we present several examples of unitaries with small Pauli $\ell_1$-norm, i.e., unitaries $U$ satisfying $\|U\|_{P,1} \leq L$ for small constant $L$.

\begin{proposition}[Multi-controlled phase unitary]
 For any $k$-qubit system ($k\le n$), define the unitary that applies a phase $\phi$ to the $\ket{1^k}$ state:
\begin{equation*}
U_{k,\phi} := I + (e^{i\phi}-1)\,\lvert 1^k\rangle\langle 1^k\rvert .
\end{equation*}
Then, $\|U_{k,\phi}\|_{P,1} \le 3$.
\end{proposition}
\begin{proof}
Using
$
\ketbra{1} = \tfrac{1}{2}(I-Z),
$
we obtain 
$$
\ketbra*{1^k} = \tfrac{1}{2^k} \sum_{S\subseteq [k]} (-1)^{\lvert S\rvert} Z_S ,
$$
where
$
Z_S := \prod_{j\in S} Z_j .
$

Therefore, the Pauli coefficients of $U_{k,\phi}$ are:
\[
\alpha_{Z_S} =
\begin{cases}
\dfrac{e^{i\phi}-1}{2^k}(-1)^{\lvert S\rvert}, & S\neq \emptyset,\\[1em]
1 + \dfrac{e^{i\phi}-1}{2^k}, & S=\emptyset .
\end{cases}
\]

Thus,
\begin{align*}
\|U_{k,\phi}\|_{P,1}
&= \left|1+\frac{e^{i\phi}-1}{2^k}\right|
  + (2^k-1)\frac{|e^{i\phi}-1|}{2^k}\\
  & \le 1 + |e^{i\phi}-1|
\le 3.
\end{align*}
\end{proof}
This unitary is a natural generalization of the Toffoli gate (which corresponds to $\phi=\pi$ and $k=2$). 

\begin{proposition}[Grover diffusion operator]
  For any $n$-qubit Grover diffusion operator $D$, $\|D\|_{P,1} < 3$.
\end{proposition}
\begin{proof}
The proof is similar to the previous example. The Grover diffusion operator (up to a global sign) is
$
D := 2\ketbra{+}^{\otimes n} - I .
$
Since 
$
\ketbra{+} = \tfrac{1}{2}(I+X),
$
we have
$
\ketbra{+}^{\otimes n} = \tfrac{1}{2^n} \sum_{S\subseteq [n]} X_S ,
$
where
$
X_S := \prod_{j\in S} X_j .
$
Thus,
\begin{align*}
D
&= -I + \frac{2}{2^n} \sum_{S\subseteq [n]} X_S .
\end{align*}

The Pauli coefficients are:
\[
\alpha_I = -1 + 2^{1-n}, \qquad
\alpha_{X_S} = 2^{1-n} \quad (S\neq \emptyset).
\]
Hence,
\begin{align*}
\|D\|_{P,1}
&= \left|-1+2^{1-n}\right|
  + (2^n-1)2^{1-n} \\
&= 3 - 2^{2-n}
< 3 .
\end{align*}
\end{proof}

\begin{example}
[Stabilizer projector ]
Let $\Pi$ be the projector onto a stabilizer code space (or any stabilizer subspace). Such a projector can always be written as
\begin{equation*}
\Pi = \frac{1}{2^k} \sum_{g\in S} g ,
\end{equation*}
where $S$ is a stabilizer group of size $2^k$ consisting of Pauli strings. Thus $\Pi$ has exactly $2^k$ Pauli terms, each with coefficient magnitude $2^{-k}$, implying
$
\|\Pi\|_{P,1} = 1 .
$
\end{example}

\begin{example}
[Selective phase unitary on a stabilizer subspace]
Continuing from the previous example, let $\Pi$ be a stabilizer projector. Define the selective phase unitary
\[U = I + (e^{i\phi}-1)\Pi
\quad\text{(equivalently } U = e^{i\phi \Pi}\text{)} .\]
By the triangle inequality, \[
\|U\|_{P,1}
\le 1 + |e^{i\phi}-1|
\le 3 .
\]
This includes reflections of the form $I-2\Pi$ by taking $\phi=\pi$.
\end{example}

\section{Conclusion and Future Work}\label{sec:conclusion}
This work studies the problem of learning unitaries that are nearly sparse in the Pauli basis. We provide upper bound on the query complexity and   established an efficient algorithm for learning this class of unitaries in diamond distance. We also proposed an algorithm for recovering large Pauli coefficients. 

We then extended our results to unitaries with bounded Pauli $\ell_1$ norm and introduced a relaxed distance measure, called restricted diamond distance, for which we provided upper and lower bounds on the query complexity of learning this class of unitaries.

Lastly, we presented several examples of unitaries with small Pauli $\ell_1$ norm including Hamiltonians corresponding to Boolean functions with bounded L1 norm. 

Future directions include improving the query complexity of learning nearly sparse unitaries in diamond distance using different techniques such as Heisenberg scaling and improved methods for implementing the approximated unitaries.


\begin{thebibliography}{52}
\providecommand{\natexlab}[1]{#1}
\providecommand{\url}[1]{\texttt{#1}}
\expandafter\ifx\csname urlstyle\endcsname\relax
  \providecommand{\doi}[1]{doi: #1}\else
  \providecommand{\doi}{doi: \begingroup \urlstyle{rm}\Url}\fi

\bibitem[Abbas et~al.(2025)Abbas, Cerrato, Guti{\'e}rrez, Grinko, Mele, and
  Sinha]{abbas2025nearly}
Amira Abbas, Nunzia Cerrato, Francisco~Escudero Guti{\'e}rrez, Dmitry Grinko,
  Francesco~Anna Mele, and Pulkit Sinha.
\newblock Nearly optimal algorithms to learn sparse quantum hamiltonians in
  physically motivated distances.
\newblock \emph{arXiv preprint arXiv:2509.09813}, 2025.

\bibitem[Aizenstein et~al.(1998)Aizenstein, Blum, Khardon, Kushilevitz, Pitt,
  and Roth]{ABKKPR1998}
Howard Aizenstein, Avrim Blum, Roni Khardon, Eyal Kushilevitz, Leonard Pitt,
  and Dan Roth.
\newblock On learning read-k-satisfy-j {DNF}.
\newblock \emph{SIAM Journal on Computing}, 27\penalty0 (6):\penalty0
  1515--1530, 1998.

\bibitem[Angrisani(2025)]{angrisani2025learning}
Armando Angrisani.
\newblock Learning unitaries with quantum statistical queries.
\newblock \emph{Quantum}, 9:\penalty0 1817, 2025.

\bibitem[Anshu et~al.(2020)Anshu, Arunachalam, Kuwahara, and
  Soleimanifar]{anshu2020sample}
Anurag Anshu, Srinivasan Arunachalam, Tomotaka Kuwahara, and Mehdi
  Soleimanifar.
\newblock Sample-efficient learning of quantum many-body systems.
\newblock In \emph{2020 IEEE 61st Annual Symposium on Foundations of Computer
  Science (FOCS)}, pages 685--691. IEEE, 2020.

\bibitem[Arunachalam and de~Wolf(2017)]{Arunachalam2017}
Srinivasan Arunachalam and Ronald de~Wolf.
\newblock A survey of quantum learning theory.
\newblock \emph{arXiv:1701.06806}, 2017.

\bibitem[Arunachalam and De~Wolf(2018)]{Arunachalam2018}
Srinivasan Arunachalam and Ronald De~Wolf.
\newblock Optimal quantum sample complexity of learning algorithms.
\newblock \emph{J. Mach. Learn. Res.}, 19\penalty0 (1):\penalty0 2879–2878,
  January 2018.
\newblock ISSN 1532-4435.

\bibitem[Arunachalam et~al.(2024)Arunachalam, Dutt, and
  Escudero~Guti{\'e}rrez]{arunachalam2024testing}
Srinivasan Arunachalam, Arkopal Dutt, and Francisco Escudero~Guti{\'e}rrez.
\newblock Testing and learning structured quantum hamiltonians, 2024.
\newblock URL \url{https://arxiv.org/abs/2411.00082}.
\newblock arXiv preprint arXiv:2411.00082 [quant-ph].

\bibitem[Arunachalam et~al.(2025)Arunachalam, Dutt, and
  Escudero~Gutiérrez]{Arunachalam2025}
Srinivasan Arunachalam, Arkopal Dutt, and Francisco Escudero~Gutiérrez.
\newblock Testing and learning structured quantum hamiltonians.
\newblock In \emph{Proceedings of the 57th Annual ACM Symposium on Theory of
  Computing}, STOC ’25, pages 1263--1270. ACM, June 2025.
\newblock \doi{10.1145/3717823.3718289}.

\bibitem[Atif et~al.(2022)Atif, Heidari, and Pradhan]{Atif2022}
Touheed~Anwar Atif, Mohsen Heidari, and S.~Sandeep Pradhan.
\newblock Faithful simulation of distributed quantum measurements with
  applications in distributed rate-distortion theory.
\newblock \emph{{IEEE} Transactions on Information Theory}, 68\penalty0
  (2):\penalty0 1085--1118, feb 2022.
\newblock \doi{10.1109/tit.2021.3124976}.

\bibitem[Bao and Yao(2023)]{Bao2023}
Zongbo Bao and Penghui Yao.
\newblock On testing and learning quantum junta channels.
\newblock In Gergely Neu and Lorenzo Rosasco, editors, \emph{Proceedings of
  Thirty Sixth Conference on Learning Theory}, volume 195 of \emph{Proceedings
  of Machine Learning Research}, pages 1064--1094. PMLR, 12--15 Jul 2023.

\bibitem[Beer et~al.(2020)Beer, Bondarenko, Farrelly, Osborne, Salzmann,
  Scheiermann, and Wolf]{beer2020training}
Kerstin Beer, Dmytro Bondarenko, Terry Farrelly, Tobias~J Osborne, Robert
  Salzmann, Daniel Scheiermann, and Ramona Wolf.
\newblock Training deep quantum neural networks.
\newblock \emph{Nature communications}, 11\penalty0 (1):\penalty0 808, 2020.

\bibitem[Belovs(2015)]{belovs2015}
Aleksandrs Belovs.
\newblock Quantum algorithms for learning symmetric juntas via the adversary
  bound.
\newblock \emph{computational complexity}, 24\penalty0 (2):\penalty0 255--293,
  2015.

\bibitem[Berry et~al.(2014)Berry, Childs, Cleve, Kothari, and Somma]{Berry2014}
Dominic~W. Berry, Andrew~M. Childs, Richard Cleve, Robin Kothari, and
  Rolando~D. Somma.
\newblock Exponential improvement in precision for simulating sparse
  hamiltonians.
\newblock In \emph{Proceedings of the forty-sixth annual ACM symposium on
  Theory of computing}, STOC ’14, pages 283--292. ACM, May 2014.
\newblock \doi{10.1145/2591796.2591854}.

\bibitem[Blum et~al.(1994)Blum, Furst, Jackson, Kearns, Mansour, and
  Rudich]{Blum1994}
Avrim Blum, Merrick Furst, Jeffrey Jackson, Michael Kearns, Yishay Mansour, and
  Steven Rudich.
\newblock Weakly learning {DNF} and characterizing statistical query learning
  using {Fourier} analysis.
\newblock In \emph{Proceedings of the twenty-sixth annual {ACM} symposium on
  Theory of computing - {STOC} 94}. {ACM} Press, 1994.
\newblock \doi{10.1145/195058.195147}.

\bibitem[Bshouty(1995)]{Bshouty95}
Nader~H. Bshouty.
\newblock Exact learning boolean function via the monotone theory.
\newblock \emph{Information and Computation}, 123\penalty0 (1):\penalty0
  146--153, 1995.

\bibitem[Bshouty et~al.(2004)Bshouty, Jackson, and Tamon]{BshoutyJT04}
Nader~H. Bshouty, Jeffrey~C. Jackson, and Christino Tamon.
\newblock More efficient {PAC}-learning of {DNF} with membership queries under
  the uniform distribution.
\newblock \emph{Journal of Computer and System Sciences}, 68\penalty0
  (1):\penalty0 205--234, 2004.

\bibitem[Chen et~al.(2024)Chen, Gong, and Ye]{Chen2024}
Sitan Chen, Weiyuan Gong, and Qi~Ye.
\newblock Optimal tradeoffs for estimating pauli observables.
\newblock In \emph{2024 IEEE 65th Annual Symposium on Foundations of Computer
  Science (FOCS)}, pages 1086--1105, 2024.
\newblock \doi{10.1109/FOCS61266.2024.00072}.

\bibitem[Chen et~al.(2022)Chen, Nadimpalli, and Yuen]{Chen2022junta}
Thomas Chen, Shivam Nadimpalli, and Henry Yuen.
\newblock Testing and learning quantum juntas nearly optimally.
\newblock \emph{arXiv:2207.05898}, July 2022.

\bibitem[Childs and Wiebe(2012)]{childs2012hamiltonian}
Andrew~M Childs and Nathan Wiebe.
\newblock Hamiltonian simulation using linear combinations of unitary
  operations.
\newblock \emph{arXiv preprint arXiv:1202.5822}, 2012.

\bibitem[Feldman(2007)]{Feldman07}
Vitaly Feldman.
\newblock Attribute-efficient and non-adaptive learning of parities and {DNF}
  expressions.
\newblock \emph{Journal of Machine Learning Research}, 8:\penalty0 1431--1460,
  2007.

\bibitem[Feldman(2012)]{Feldman2012}
Vitaly Feldman.
\newblock Learning {DNF} expressions from {Fourier} spectrum.
\newblock In Shie Mannor, Nathan Srebro, and Robert~C. Williamson, editors,
  \emph{Proceedings of the 25th Annual Conference on Learning Theory},
  volume~23 of \emph{Proceedings of Machine Learning Research}, pages
  17.1--17.19, 2012.

\bibitem[Goldreich and Levin(1989)]{Goldreich1989}
O.~Goldreich and L.~A. Levin.
\newblock A hard-core predicate for all one-way functions.
\newblock In \emph{Proceedings of the twenty-first annual ACM symposium on
  Theory of computing}, pages 25--32. ACM Press, 1989.

\bibitem[Grewal and Liang(2025)]{Grewal2025}
Sabee Grewal and Daniel Liang.
\newblock Query-optimal estimation of unitary channels via pauli
  dimensionality.
\newblock September 2025.
\newblock \doi{10.48550/ARXIV.2510.00168}.

\bibitem[Gutoski and Johnston(2014)]{gutoski2014process}
Gus Gutoski and Nathaniel Johnston.
\newblock Process tomography for unitary quantum channels.
\newblock \emph{Journal of Mathematical Physics}, 55\penalty0 (3), 2014.

\bibitem[Haah et~al.(2016)Haah, Harrow, Ji, Wu, and Yu]{Haah2016}
Jeongwan Haah, Aram~W. Harrow, Zhengfeng Ji, Xiaodi Wu, and Nengkun Yu.
\newblock Sample-optimal tomography of quantum states.
\newblock In \emph{Proceedings of the forty-eighth annual {ACM} symposium on
  Theory of Computing}. {ACM}, jun 2016.
\newblock \doi{10.1145/2897518.2897585}.

\bibitem[Haah et~al.(2023{\natexlab{a}})Haah, Kothari, O’Donnell, and
  Tang]{Haah2023}
Jeongwan Haah, Robin Kothari, Ryan O’Donnell, and Ewin Tang.
\newblock Query-optimal estimation of unitary channels in diamond distance.
\newblock In \emph{2023 IEEE 64th Annual Symposium on Foundations of Computer
  Science (FOCS)}, pages 363--390. IEEE, November 2023{\natexlab{a}}.
\newblock \doi{10.1109/focs57990.2023.00028}.

\bibitem[Haah et~al.(2023{\natexlab{b}})Haah, Kothari, O’Donnell, and
  Tang]{haah2023query}
Jeongwan Haah, Robin Kothari, Ryan O’Donnell, and Ewin Tang.
\newblock Query-optimal estimation of unitary channels in diamond distance.
\newblock In \emph{2023 IEEE 64th Annual Symposium on Foundations of Computer
  Science (FOCS)}, pages 363--390. IEEE, 2023{\natexlab{b}}.

\bibitem[Heidari and Khardon(2025)]{HeidariK2025}
M.~Heidari and R.~Khardon.
\newblock Learning {DNF} through generalized {Fourier} representations.
\newblock In \emph{Proceedings of the Conference on Learning Theory}, 2025.
\newblock Full paper available as arXiv preprint 2506.01075.

\bibitem[Heidari and Szpankowski(2023)]{Heidari2023a}
Mohsen Heidari and Wojciech Szpankowski.
\newblock Learning k-qubit quantum operators via pauli decomposition.
\newblock In Francisco Ruiz, Jennifer Dy, and Jan-Willem van~de Meent, editors,
  \emph{Proceedings of The 26th International Conference on Artificial
  Intelligence and Statistics}, volume 206 of \emph{Proceedings of Machine
  Learning Research}, pages 490--504. PMLR, 25--27 Apr 2023.
\newblock URL \url{https://proceedings.mlr.press/v206/heidari23a.html}.

\bibitem[Heidari and Szpankowski(2024)]{Heidari2024}
Mohsen Heidari and Wojciech Szpankowski.
\newblock New bounds on quantum sample complexity of measurement classes.
\newblock In \emph{2024 IEEE International Symposium on Information Theory
  (ISIT)}, pages 1515--1520. IEEE, July 2024.
\newblock \doi{10.1109/isit57864.2024.10619538}.

\bibitem[Heidari et~al.(2021)Heidari, Padakandla, and
  Szpankowski]{HeidariQuantum2021}
Mohsen Heidari, Arun Padakandla, and Wojciech Szpankowski.
\newblock A theoretical framework for learning from quantum data.
\newblock In \emph{2021 {IEEE} International Symposium on Information Theory
  ({ISIT})}. {IEEE}, jul 2021.
\newblock \doi{10.1109/isit45174.2021.9517721}.

\bibitem[Hellerstein et~al.(2012)Hellerstein, Kletenik, Sellie, and
  Servedio]{HellersteinKSS12}
Lisa Hellerstein, Devorah Kletenik, Linda Sellie, and Rocco~A. Servedio.
\newblock Tight bounds on proper equivalence query learning of {DNF}.
\newblock In \emph{The 25th Annual Conference on Learning Theory}, pages
  31.1--31.18, 2012.

\bibitem[Hu et~al.(2025)Hu, Ma, Gong, Ye, Tong, Flammia, and
  Yelin]{hu2025ansatz}
Hong-Ye Hu, Muzhou Ma, Weiyuan Gong, Qi~Ye, Yu~Tong, Steven~T Flammia, and
  Susanne~F Yelin.
\newblock Ansatz-free hamiltonian learning with heisenberg-limited scaling.
\newblock \emph{arXiv preprint arXiv:2502.11900}, 2025.

\bibitem[Huang et~al.(2021)Huang, Kueng, and Preskill]{Huang2021}
H.-Y. Huang, R.~Kueng, and J.~Preskill.
\newblock Information-theoretic bounds on quantum advantage for learning.
\newblock \emph{Physical Review Letters}, 127:\penalty0 030503, 2021.
\newblock \doi{10.1103/PhysRevLett.127.030503}.

\bibitem[Huang et~al.(2020)Huang, Kueng, and Preskill]{Huang2020}
Hsin-Yuan Huang, Richard Kueng, and John Preskill.
\newblock Predicting many properties of a quantum system from very few
  measurements.
\newblock \emph{Nature Physics 16, 1050--1057 (2020)}, February 2020.
\newblock \doi{10.1038/s41567-020-0932-7}.

\bibitem[Jackson(1997)]{Jackson1997}
Jeffrey~C Jackson.
\newblock An efficient membership-query algorithm for learning {DNF} with
  respect to the uniform distribution.
\newblock \emph{Journal of Computer and System Sciences}, 55\penalty0
  (3):\penalty0 414--440, dec 1997.
\newblock \doi{10.1006/jcss.1997.1533}.

\bibitem[Kalai et~al.(2009)Kalai, Samorodnitsky, and Teng]{Kalai2009}
Adam~Tauman Kalai, Alex Samorodnitsky, and Shang-Hua Teng.
\newblock Learning and smoothed analysis.
\newblock In \emph{2009 50th Annual IEEE Symposium on Foundations of Computer
  Science}, pages 395--404. IEEE, October 2009.

\bibitem[Kearns et~al.(1994)Kearns, Schapire, and Sellie]{Kearns1994}
Michael~J. Kearns, Robert~E. Schapire, and Linda~M. Sellie.
\newblock Toward efficient agnostic learning.
\newblock \emph{Machine Learning}, 17\penalty0 (2-3):\penalty0 115--141, 1994.

\bibitem[Khardon(1994)]{Khardon94}
Roni Khardon.
\newblock On using the {Fourier} transform to learn disjoint {DNF}.
\newblock \emph{Information Processing Letters}, 49\penalty0 (5):\penalty0
  219--222, 1994.

\bibitem[Kiani et~al.(2020)Kiani, Lloyd, and Maity]{kiani2020learning}
Bobak~Toussi Kiani, Seth Lloyd, and Reevu Maity.
\newblock Learning unitaries by gradient descent.
\newblock \emph{arXiv preprint arXiv:2001.11897}, 2020.

\bibitem[King et~al.()King, Gosset, Kothari, and Babbush]{King2025}
Robbie King, David Gosset, Robin Kothari, and Ryan Babbush.
\newblock \emph{Triply efficient shadow tomography}, pages 914--946.
\newblock \doi{10.1137/1.9781611978322.27}.
\newblock URL \url{https://epubs.siam.org/doi/abs/10.1137/1.9781611978322.27}.

\bibitem[Kothari(2014)]{Kothari2014}
Robin Kothari.
\newblock \emph{Efficient algorithms in quantum query complexity}.
\newblock PhD thesis, University of Waterloo, 2014.
\newblock URL \url{http://hdl.handle.net/10012/8625}.

\bibitem[Kushilevitz(1996)]{Kushilevitz96}
Eyal Kushilevitz.
\newblock A simple algorithm for learning {O}(log n)-term {DNF}.
\newblock In \emph{Proceedings of the Ninth Annual Conference on Computational
  Learning Theory}, pages 266--269, 1996.

\bibitem[Kushilevitz and Mansour(1993)]{KM1993}
Eyal Kushilevitz and Yishay Mansour.
\newblock Learning decision trees using the {Fourier} spectrum.
\newblock \emph{SIAM Journal on Computing}, 22\penalty0 (6):\penalty0
  1331--1348, December 1993.

\bibitem[Levy et~al.(2024)Levy, Luo, and Clark]{levy2024classical}
Ryan Levy, Di~Luo, and Bryan~K Clark.
\newblock Classical shadows for quantum process tomography on nearterm quantum
  computers.
\newblock \emph{Physical Review Research}, 6\penalty0 (1):\penalty0 013029,
  2024.

\bibitem[Montanaro and Osborne(2010)]{Montanaro2010}
Ashley Montanaro and Tobias~J. Osborne.
\newblock Quantum boolean functions.
\newblock \emph{arXiv:0810.2435}, October 2010.

\bibitem[M\"{o}tt\"{o}nen et~al.(2005)M\"{o}tt\"{o}nen, Vartiainen, Bergholm,
  and Salomaa]{Moettoenen2005}
Mikko M\"{o}tt\"{o}nen, Juha~J. Vartiainen, Ville Bergholm, and Martti~M.
  Salomaa.
\newblock Transformation of quantum states using uniformly controlled
  rotations.
\newblock \emph{Quantum Info. Comput.}, 5\penalty0 (6):\penalty0 467–473,
  September 2005.
\newblock ISSN 1533-7146.

\bibitem[Valiant(1984)]{Valiant1984}
L.~G. Valiant.
\newblock A theory of the learnable.
\newblock \emph{Communications of the {ACM}}, 27\penalty0 (11):\penalty0
  1134--1142, November 1984.

\bibitem[Wilde(2013)]{Wilde2013}
Mark Wilde.
\newblock \emph{Quantum information theory}.
\newblock Cambridge University Press, Cambridge, UK, 2013.
\newblock ISBN 9781139525343.

\bibitem[Wilde et~al.(2012)Wilde, Hayden, Buscemi, and
  Hsieh]{wilde2012information}
Mark~M. Wilde, Patrick Hayden, Francesco Buscemi, and Min-Hsiu Hsieh.
\newblock The information-theoretic costs of simulating quantum measurements.
\newblock \emph{arXiv preprint arXiv:1206.4121}, 2012.
\newblock URL \url{https://arxiv.org/abs/1206.4121}.
\newblock v2, revised version.

\bibitem[Winter(2004)]{winter2004extrinsic}
Andreas Winter.
\newblock ‘‘extrinsic’’and ‘‘intrinsic’’data in quantum
  measurements: Asymptotic convex decomposition of positive operator valued
  measures.
\newblock \emph{Communications in mathematical physics}, 244\penalty0
  (1):\penalty0 157--185, 2004.

\bibitem[Zhao et~al.(2024)Zhao, Lewis, Kannan, Quek, Huang, and Caro]{Zhao2024}
Haimeng Zhao, Laura Lewis, Ishaan Kannan, Yihui Quek, Hsin-Yuan Huang, and
  Matthias~C. Caro.
\newblock Learning quantum states and unitaries of bounded gate complexity.
\newblock \emph{PRX Quantum}, 5\penalty0 (4):\penalty0 040306, October 2024.
\newblock ISSN 2691-3399.
\newblock \doi{10.1103/prxquantum.5.040306}.

\end{thebibliography}
\end{document}